\documentclass[a4paper]{article}
\usepackage{amsthm}
\usepackage{amsmath}
\usepackage{amsfonts}
\usepackage{amssymb}

\usepackage[all]{xy}

\allowdisplaybreaks
\usepackage{mathtools}
\usepackage{stmaryrd}

\usepackage{multicol}

\usepackage[shortlabels]{enumitem}
\setlist[enumerate]{leftmargin=*,align=left,labelindent=\parindent}

\makeatletter
\newcommand{\myitem}[1][]{%
\item[(#1)]\protected@edef\@currentlabel{\textup{#1}}\ignorespaces%
}
\makeatother

\usepackage{hyperref}

\newtheorem{theorem}{Theorem}[section]
\newtheorem{lemma}[theorem]{Lemma}
\newtheorem{proposition}[theorem]{Proposition}
\newtheorem{corollary}[theorem]{Corollary}

\theoremstyle{definition}
\newtheorem{definition}[theorem]{Definition}

\theoremstyle{remark}
\newtheorem{remark}[theorem]{Remark}
\newtheorem{notation}[theorem]{Notation}

\numberwithin{equation}{section}

\newcommand{\ci}{\mathtt{i}}
\newcommand{\ck}{\mathtt{k}}
\newcommand{\cs}{\mathtt{s}}
\newcommand{\ce}{\mathtt{e}}

\newcommand{\cf}{\mathtt{f}}
\newcommand{\ct}{\mathtt{t}}

\newcommand{\vX}{\mathtt{X}}
\newcommand{\vx}{\mathtt{x}}
\newcommand{\vy}{\mathtt{y}}
\newcommand{\vz}{\mathtt{z}}
\newcommand{\Unit}{\mathit{I}}
\newcommand{\terminal}{\mathtt{T}}

\newcommand{\pa}{\star}
\newcommand{\qa}{\bullet}
\newcommand{\qna}{\ast}

\newcommand{\LambdaAbst}[1]{\Lambda(#1)}
\newcommand{\pair}[1]{\langle #1 \rangle}

\newcommand{\Cat}[1]{\mathbb{#1}}
\newcommand{\Nat}{\mathbb{N}}
\newcommand{\id}{\mathrm{id}}
\newcommand{\FV}{\mathrm{FV}}

\title{Reflexive combinatory algebras}

\usepackage{authblk}

\author{Marlou M. Gijzen\thanks{\texttt{marlou.gijzen@gmail.com}}}
\author{Hajime Ishihara\thanks{\texttt{ishihara@jaist.ac.jp}}}
\author{Tatsuji Kawai\thanks{\texttt{tatsuji.kawai@jaist.ac.jp}
}}
\affil{Japan Advanced Institute of Science and Technology\authorcr
1-1 Asahidai, Nomi, Ishikawa 923-1292, Japan}

\date{}

\begin{document}
\maketitle

\begin{abstract}
  We introduce the notion of reflexivity for combinatory algebras. 
  Reflexivity can be thought of as an equational counterpart of
  the Meyer--Scott axiom of combinatory models, which indeed allows us to
  characterise an equationally definable counterpart of combinatory
  models. This new structure, called strongly reflexive combinatory
  algebra, admits a finite axiomatisation with seven closed equations,
  and the structure is shown to be exactly the retract of combinatory
  models.
  Lambda algebras can be characterised as strongly reflexive
  combinatory algebras which are stable. Moreover, there is a canonical
  construction of a lambda algebra from a strongly reflexive
  combinatory algebra. The resulting axiomatisation of lambda algebras
  by the seven axioms for strong reflexivity
  together with those for stability is shown to correspond to the
  axiomatisation of lambda algebras due to Selinger [J.\ Funct.\
  Programming, 12(6), 549--566, 2002].
  \bigskip
  %
  % [TODO] strongly reflexive combinatory pre-model? or strongly
  % reflexive combinatory pre-model?
  %

\medskip

\end{abstract}

%%%%%%%%%%%%%%%%%%%%%%%%%%%%%%%%%%%%%%%%%%%%%%%%
\section{Introduction}\label{sec:Introduction}
%%%%%%%%%%%%%%%%%%%%%%%%%%%%%%%%%%%%%%%%%%%%%%%%
The paper is written in the hope that the study of structures more
general than the established notion of models of the lambda calculus
(i.e., lambda algebras) should lead to a better understanding of the
properties of these models. To this end, we introduce the notion of
\emph{reflexivity} in a general setting for combinatory algebras,
and show how this property relates to models of the lambda calculus.

Recall that a lambda algebra is a combinatory algebra which satisfies
Curry's five closed axioms (cf.\ Definition \ref{def:LambdaAlg}).
Apart from this succinct axiomatisation, they have nice structural
properties: lambda algebras are exactly the retracts of lambda models
\cite{MeyerLambda}, the first-order models of lambda calculus
characterised as lambda algebras satisfying the Meyer--Scott axiom
(cf.\ Definition~\ref{def:CombMod}). Moreover, lambda algebras
correspond to the reflexive objects in cartesian closed categories
\cite{KoymansModelOfLambda,ScottLambda}. 
However, the choice of axioms for lambda algebras by Curry is elusive
and looks arbitrary~\cite{BarendregtLambda,MeyerLambda,
Selinger02,KoymansModelOfLambda}.%
\footnote{See also Lambek~\cite{Lambek}, Freyd~\cite{FreydCombinator},
Hindley and Seldin~\cite[Chapter 8,
8B]{HindleySeldinLambdaCombinators} for discussions on the finite
axiomatisation of Curry algebras (namely, the models of
$\lambda_{\beta\eta}$).}

On the other hand, lambda models can be characterised independently
from lambda algebras as combinatory models which are stable \cite{MeyerLambda}.
Combinatory models can be characterised simply as combinatory algebras
satisfying the Meyer--Scott axiom:
  \begin{equation*}
    \left[ \forall c \in \underline{A} \left( ac = bc \right) \right]
    \implies {\ce a = \ce b}
  \end{equation*}
where $\ce$ is a distinguished constant (often defined as $\cs
(\ck \ci)$).
Combinatory models are sufficient for interpreting lambda
calculus. Moreover, there is a canonical way of stabilising a
combinatory model to obtain a lambda model. In this sense, combinatory
models encapsulate the essence of lambda models.

The observation above allows us to identify \emph{reflexivity} as a
fundamental property of combinatory algebras, which comes into play as follows:
the combinatory completeness implies that there is a surjection
$\varphi \colon a \mapsto a \vx$ from a combinatory algebra $A$ to its
polynomial algebra $A[\vx]$, whose section is provided by any choice
of a defined $\lambda$-abstraction.  The problem is that familiar
abstraction mechanisms fail to respect the equality on the
polynomials: the question is thus under what conditions this
succeeds. Since $\varphi$ is surjective, it induces an
equivalence relation $\sim_{A}$ on $A$ which makes $A
\slash{\sim_{A}}$ a combinatory algebra isomorphic to $A[\vx]$. The
key observation is that the relation $\sim_{A}$ is generated by finite
schemas of equations on the elements of $A$. The above problem is then
reduced to the condition, called reflexivity, that the constant
$\ce$ preserves these equations, namely
\begin{equation}
  \label{eq:reflexivity}
  a \sim_{A} b \implies \ce a = \ce b,
\end{equation}
or in terms of polynomial algebra, 
$a \vx = b \vx \implies \ce a = \ce b$.
Reflexivity yields seven simple universal sentences on $A$, which are
equivalent to the requirement that  an \emph{alternative} choice of
lambda abstraction, denoted $\lambda^\dag$, is a well-defined
operation on the polynomial algebra with one indeterminate. Then,
taking $\lambda^\dag$-closures of these seven sentences yields seven
closed equations; these equations correspond to the requirement that
the polynomial algebra be reflexive, or equivalently, that
$\lambda^\dag$ be well-defined on polynomial algebras with any finite
numbers of indeterminates.

We call a combinatory algebra satisfying the seven closed equations as
mentioned above a \emph{strongly reflexive combinatory algebra}.
The
class of strongly reflexive combinatory algebras provides the
algebraic (i.e., equational) counterpart to the notion of
combinatory models in that they are precisely the retracts of
combinatory models. Indeed, the condition \eqref{eq:reflexivity} can
be considered as an equational counterpart of the Meyer--Scott axiom.
On the other hand, combinatory models can be characterised as strongly
reflexive combinatory algebras satisfying the Meyer--Scott axiom.
The relation between strongly reflexive combinatory algebras and
lambda algebras can then be captured by
stability~\cite[5.6.4~(ii)]{BarendregtLambda}.  Moreover,
every strongly reflexive combinatory algebra can be made into a stable
one (i.e., a lambda algebra) with an appropriate choice of constants.
This passage to a lambda algebra also manifests itself in another
form: we can associate a cartesian closed monoid (and thus a cartesian
closed category with a reflexive object) to a strongly reflexive
combinatory algebra, from which the above mentioned lambda algebra is
obtained. Furthermore, the resulting axiomatisation of lambda algebras
with the seven closed equations and the axiom of stability naturally
corresponds to the axiomatisation of lambda algebras due to Selinger~\cite{Selinger02}. 
Thus, the notion of strongly reflexive combinatory algebras serves as a
common generalisation of those of combinatory models and lambda
algebras, which moreover is finitely axiomatisable (see Figure~\ref{fig:IntroDiagram}).

Throughout this paper, we work with \emph{combinatory
pre-models}, combinatory algebras extended with distinguished elements
$\ci$ and $\ce$. Of course, these elements can be defined in terms of
elements $\ck$ and $\cs$ as $\ci = \cs \ck \ck$ and $\ce = \cs (\ck
\ci)$. Nevertheless, the inclusion of $\ce$ as a primitive, in
particular, may be justified given its fundamental role in reflexivity.  
Accordingly, most of the notions for combinatory algebras mentioned
above (such as polynomial algebra, reflexivity, strong reflexivity,
stability) will be introduced for combinatory pre-models.  The
exceptions are lambda algebras and lambda models, which
are defined with respect to the conventional notion of combinatory
algebras with  $\ck$ and $\cs$ as the only primitive constants.
 
\begin{figure}[tb]
  \small
  \centering
 \[
   \xymatrix@H=+27pt@C=-28pt@L=6pt{
     & \txt{\bf Combinatory algebras}
     \ar[d]_{\txt{seven sentences}}
     \ar `r[rddd]  [rddd]^{\txt{Curry's axioms}}
     &  \\
     & \txt{\bf Reflexive combinatory algebras}
     \ar[d]_{\txt{$\lambda^\dag$-closures of seven sentences}}
     & \\
     & \txt{\bf Strongly reflexive combinatory algebras}
     \ar @/_1pc/ [dl]_(.7){\txt{Meyer--Scott axiom}}
     \ar[dr]^(.6){\txt{stability}}
     &\\
      \txt{\bf Combinatory models}
     \ar @/_1pc/ [ur]_(.6){\txt{retracts}}
     \ar [dr]_(.5){\txt{stability}}
     & &
     \txt{\bf Lambda algebras} 
     \ar @/^1pc/ [dl]^(.3){\txt{Meyer--Scott axiom}}
     \\
     & \txt{\bf Lambda models} 
     \ar @/^1pc/ [ur]^(.4){\txt{retracts}} &   
    }
 \]
 \caption{Relationship between various structures} \label{fig:IntroDiagram}
\end{figure}
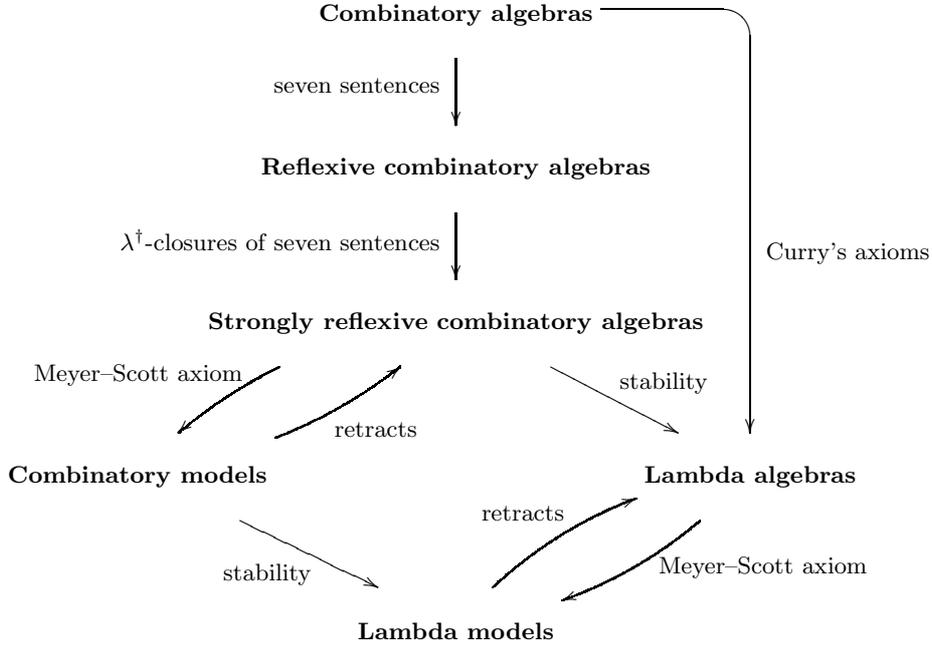

\paragraph{Organisation}
In Section~\ref{sec:CombPreModel}, we introduce the notion of
combinatory pre-models and establish some elementary properties of polynomial algebras. We
also introduce an alternative representation of a polynomial algebra
without using indeterminates.
In Section~\ref{sec:Reflexivity}, we introduce the notion of
reflexivity for combinatory pre-models, and show that this can
be characterised by seven simple universal sentences.  We also
introduce a new abstraction mechanism, denoted $\lambda^\dag$,  in
terms of which reflexivity can be rephrased.
In Section~\ref{sec:AlgCombModel}, we introduce the notion of strong
reflexivity for combinatory pre-models. We show that the class of strongly
reflexive combinatory pre-models is axiomatisable with a finite set of
equations which can be obtained by taking $\lambda^\dag$-closures of
seven axioms of reflexivity. We then show that strongly reflexive
combinatory pre-models are exactly the retracts of combinatory models.
In Section~\ref{sec:CCCMoid}, we generalise the construction of 
a cartesian closed monoid from a lambda algebra
to the setting of strongly reflexive combinatory pre-models.
In Section~\ref{sec:Stability}, 
we introduce the notion of stability for strongly reflexive
combinatory pre-models
and characterise lambda algebras as
stable strongly reflexive combinatory pre-models.
We then clarify how this characterisation of lambda algebras
corresponds to that of Selinger.
%

%%%%%%%%%%%%%%%%%%%%%%%%%%%%%%%%%%%%%%%%%%%%%%%%
\section{Combinatory pre-models}\label{sec:CombPreModel}
%%%%%%%%%%%%%%%%%%%%%%%%%%%%%%%%%%%%%%%%%%%%%%%%
We begin with a preliminary on quotients of applicative structures.
\begin{definition}\label{as}
An \emph{applicative structure} is a pair
\(
 A = (\underline{A},\mathbin{\cdot})
\)
where $\underline{A}$ is a set and $\mathbin{\cdot}$ is a binary
operation on
$\underline{A}$, called an \emph{application}.
The application $a \cdot b$ is often written as $(ab)$, and
parentheses are omitted following the convention of association to
the left.
 
A \emph{homomorphism} between applicative structures
\(
 A = (\underline{A},\mathbin{\cdot_A})
\)
and
\(
 B = (\underline{B},\mathbin{\cdot_B})
\)
is a function $f \colon \underline{A} \to \underline{B}$ such that
\[
  f(a \cdot_A b) =
  f(a) \cdot_B f(b)
\]
for each $a, b \in \underline{A}$.
\end{definition}

In what follows, we fix an applicative structure 
 $A = (\underline{A},\mathbin{\cdot})$.
\begin{definition}\label{def:Congruence}
  A \emph{congruence relation} (or simply a congruence) on $A$ is an equivalence relation $\sim$ on
  $\underline{A}$ such that 
  \[
    a \sim b \;\text{and}\; c \sim d \implies a \cdot c \sim b \cdot
    d
  \]
  for each $a,b,c,d \in \underline{A}$.
\end{definition}
\begin{notation}
  If $\sim$ is a congruence on $A$, we often write $a$ for the
  equivalence class $[a]_{\sim}$ of
  $\sim$ whenever it is clear from
  the context. In this case, we write $a \sim b$ for $[a]_{\sim} =
  [b]_{\sim}$. The  convention also applies to the other congruence
  relations in this paper.
\end{notation}

\begin{definition}
  Let $\sim$ be a congruence on $A$. The \emph{quotient} of $A$ by
  $\sim$ is an applicative structure
  \[
    A/{\sim} =
    (\underline{A}/{\sim}, \qa)
  \]
  where $a \qa b = a \cdot b$. There is a natural homomorphism
  $\pi_{\sim} \colon A \to A/{\sim}$ defined by $\pi_{\sim}(a) = a$.
\end{definition}

For any binary relation $R$ on $\underline{A}$, there is a smallest
congruence $\sim_{R}$ on $A$ containing $R$, which is inductively
generated by the following rules:
\begin{enumerate}
  \item if $a \mathrel{R} b$, then $a \sim_{R} b$,
\item $a \sim_{R} a$,
\item if $a \sim_{R} b$, then $b \sim_{R} a$,
\item if $a \sim_{R} b$ and $b \sim_{R} c$, then $a \sim_{R} c$,
\item if $a \sim_{R} b$ and $c \sim_{R} d$, then $a \cdot c
  \sim_{R} b \cdot d$,
\end{enumerate}
where $a, b, c, d \in \underline{A}$. 

\begin{proposition}
  \label{prop:AppUniExt}
  Let $R$ be a binary relation on $\underline{A}$, and let $f \colon A
  \to B$ be a homomorphism of applicative structures such that $f(a) = f(b)$  for each $a
  \mathrel{R} b$. Then, there exists a unique homomorphism
  $\widetilde{f} \colon A /{\sim_R} \to B$ such that 
  $\widetilde{f} \circ \pi_{\sim_{R}} = f$.
\end{proposition}
\begin{proof}
  Define $\widetilde{f}$ by $\widetilde{f}(a) = f(a)$. 
  By induction on $\sim_{R}$, one can show that ${a \sim_{R} b}$ implies
  $f(a) = f(b)$ for each $a,b \in \underline{A}$. The fact that
  $\widetilde{f}$ is a homomorphism is clear. 
  The uniqueness of $\widetilde{f}$ follows from the fact that  $\pi_{\sim_{R}}$ 
  is surjective.
\end{proof}

Throughout the paper, we work with the following notion of 
combinatory algebras where constants $\ci$ and $\ce$ are
given as primitives.
\begin{definition}\label{cpmod}
A \emph{combinatory pre-model} is a structure
\(
 A = (\underline{A},\mathbin{\cdot},\ck,\cs,\ci,\ce)
\)
where
\(
 (\underline{A},\mathbin{\cdot})
\)
is an applicative structure and $\ck, \cs$, $\ci, \ce$ are elements of
$\underline{A}$ such that
\begin{enumerate}
  \item\label{Ax:K} $ \ck a b = a,$ 

  \item \label{Ax:S} $ \cs a b c = a c (b c),$
    
  \item\label{Ax:I} $ \ci a = a,$ 
    
  \item\label{Ax:E} $ \ce a b = a b$ 
\end{enumerate}
for each $a, b, c \in \underline{A}$.
The reduct 
$(\underline{A},\mathbin{\cdot},\ck,\cs)$ is called a
\emph{combinatory algebra}.
 
A \emph{homomorphism} between combinatory pre-models
 $
   A = (\underline{A},\mathbin{\cdot_A},\ck_A,\cs_A,\ci_A,\ce_A)
 $
and
 $
 B = (\underline{B},\mathbin{\cdot_B},\ck_B,\cs_B,\ci_B,\ce_B)
 $
is a homomorphism between applicative structures
$
(\underline{A},\mathbin{\cdot_A})
$
and
$
(\underline{B},\mathbin{\cdot_B})
$
such that
 $
 f(\ck_A) = \ck_B,
 $
 $
 f(\cs_A) = \cs_B,
 $
 $
 f(\ci_A) = \ci_B,
 $
and
 $
 f(\ce_A) = \ce_B.
 $
Combinatory pre-models $A$ and $B$ are \emph{isomorphic} if there exists
a bijective homomorphism between $A$ and $B$. Homomorphisms and
isomorphisms between
combinatory algebras are defined similarly.

In what follows, homomorphisms mean homomorphisms between combinatory
pre-models unless otherwise noted. 
\end{definition}
\begin{notation}
  \label{not:CPM}
  We often use a combinatory pre-model denoted by the letter $A$.
  Unless otherwise noted, we assume that $A$ has the underlying structure
  \(
   A = (\underline{A},\mathbin{\cdot},\ck, \cs,\ci,\ce).
  \)
\end{notation}

We recall the construction of a polynomial algebra, and establish its
basic properties.
\begin{definition}\label{cr}
Let $S$ be a set. The set $\mathcal{T}(S)$ of \emph{terms over}
$S$ is inductively generated by the following rules:
\begin{enumerate}
\item $a \in \mathcal{T}(S)$ for each $a \in S$,
\item if $t, u \in \mathcal{T}(S)$, then $(t,u) \in \mathcal{T}(S)$.
\end{enumerate}
Note that $\mathcal{T}(S)$ is a free applicative structure
$( \mathcal{T}(S), \cdot )$
over $S$ where $t \cdot u = (t, u)$.
\end{definition}

In the rest of this section, we work over a fixed  combinatory pre-model
$A$.
We assume that a countably infinite set $\vX = \left\{ \vx_{i} \mid i \geq 1
\right\}$ of distinct indeterminates is given.
For each term $t \in \mathcal{T}(\vX + \underline{A})$, $\FV(t)$
denotes the set of indeterminates that occur in $t$.
As usual, $t$ is said to be \emph{closed} if $\FV(t) = \emptyset$.
\begin{definition}
\label{def:PolyAlg}
Let $\approx_\vX$ be
the congruence relation on $\mathcal{T}(\vX + \underline{A})$ 
generated from the following basic relation:
\begin{enumerate}
\item\label{eq:polyK} $((\ck,s),t) \approx_\vX s$,
\item\label{eq:polyS} $(((\cs,s),t),u) \approx_\vX ((s,u),(t,u))$,
\item\label{eq:polyI} $(\ci,s) \approx_\vX s$,
\item\label{eq:polyE} $((\ce,s),t) \approx_\vX (s,t)$,
\item\label{eq:polyInj} $(a,b) \approx_\vX ab$,
\end{enumerate}
where $a, b \in \underline{A}$ and $s,t,u \in
\mathcal{T}(\vX + \underline{A})$.
The \emph{polynomial algebra} $A[\vX]$ over $A$ is a combinatory pre-model
\begin{equation}
  \label{eq:PolyAlg}
  A[\vX] = (\underline{A[\vX]}, \mathbin{\pa},\ck, \cs, \ci,\ce)
\end{equation}
where
$(\underline{A[\vX]}, \pa)$ is the quotient of 
$\mathcal{T}(\{\vX\}+\underline{A})$ with respect to $\approx_\vX$.
There is a homomorphism $\sigma_{A} \colon A \to A[\vX]$ defined by
$\sigma_{A}(a) = a$. 

Similarly, for each $n \in \Nat$, the congruence relation
$\approx_{\vx_1 \dots \vx_n}$ 
on $\mathcal{T}(\left\{ \vx_1,\dots,\vx_n \right\} + \underline{A})$ is generated from the five basic equations
above. The polynomial algebra $A[\vx_1,\dots,\vx_n]$
over $A$ in indeterminates $\vx_1,\dots,\vx_n$ is then defined as in
\eqref{eq:PolyAlg}, whose underlying set is the quotient of 
$\mathcal{T}(\left\{ \vx_1,\dots,\vx_n \right\} + \underline{A})$
with respect to $\approx_{\vx_1\dots\vx_n}$.
Let $\eta_{n} \colon A \to A[\vx_1,\dots,\vx_n]$ be the
homomorphism defined by  $\eta_n(a) = a$.
\end{definition}
\begin{notation}
  \leavevmode
  \begin{enumerate}
    \item We sometimes write $\eta_{A} \colon A \to A[\vx_1]$ for
      $\eta_{1} \colon A \to A[\vx_1]$.

    \item We use $\vx,\vy,\vz$ for $\vx_1,\vx_2,\vx_3$, respectively;
      thus, $A[\vx] = A[\vx_1]$, $A[\vx,\vy] = A[\vx_1,\vx_2]$, and
      $A[\vx,\vy,\vz] = A[\vx_1,\vx_2,\vx_3]$.  However, we sometimes
      use $\vx$ for an arbitrary element of $\vX$ (cf.\ Definition
      \ref{def:LambdaAst} and Definition \ref{def:LambdaDag}).  The
      meaning of $\vx$ should be clear from the context.

    \item Terms are often written without parentheses and commas;
      e.g., $(t,s)$ will be written simply as $ts$. In most cases, the
      reader should be able to reconstruct the original terms
      following the convention of association to the left. However,
      there still remain some ambiguities; e.g., it is not clear
      whether $ab$ for $a,b \in \underline{A}$ denotes $(a,b)$ or $ab
      \in \underline{A}$.  In practice, this kind of distinction does
      not matter as terms are usually considered up to equality of
      polynomial algebras.
  \end{enumerate}
\end{notation}

We recall some standard properties of polynomial algebras.
\begin{lemma}
  \label{prop:UniversalPolyAlg}
  Let  $ B = (\underline{B},\mathbin{\cdot_B},\ck_B,\cs_B,\ci_B,\ce_B)$ be a
  combinatory pre-model, and let $n \geq 1$. 
  For each homomorphism $f \colon A \to B$ and elements
  $b_{1},\dots,b_{n} \in \underline{B}$, there exists a
  unique homomorphism $\overline{f} \colon A[\vx_1,\dots,\vx_n] \to B$ such that
  $\overline{f} \circ \eta_{n} = f$ and $\overline{f}(\vx_i) = b_i$
  for each $i \leq n$.
\end{lemma}
\begin{proof}
  Let $f \colon A \to B$ be a homomorphism and $b_1,\dots,b_n \in
  \underline{B}$.
  By the freeness of $\mathcal{T}(\left\{ \vx_1,\dots,\vx_n
  \right\}+\underline{A})$, $f$
  uniquely extends to a homomorphism of applicative structures
  $\widetilde{f} \colon \mathcal{T}(\left\{ \vx_1,\dots,\vx_n
  \right\}+\underline{A}) \to B$
  such that $\widetilde{f}(\vx_i) = b_i$ for $i \leq n$,
  $\widetilde{f}(a) = f(a)$ for $a \in \underline{A}$, 
  and $\widetilde{f}((t,u)) = \widetilde{f}(t) \mathbin{\cdot_B} \widetilde{f}(u)$.
  Since $f$ is a homomorphism of combinatory pre-models, 
  $\widetilde{f}$ satisfies the assumption of Proposition
  \ref{prop:AppUniExt}; hence it uniquely extends to a homomorphism
  $\overline{f} \colon A[\vx_1,\dots,\vx_n] \to B$. 
  Then, we have $\bar{f}(\vx_i) =
  \widetilde{f}(\vx_i) = b_i$ and $\bar{f}(\eta_A(a)) = \widetilde{f}(a) =
  f(a)$ for each $a \in \underline{A}$.
\end{proof}
Similarly, we have the following for $A[\vX]$.
\begin{lemma}
  \label{prop:UnivPropPolynomialAlgX}
  Let  $ B =
  (\underline{B},\mathbin{\cdot_B},\ck_B,\cs_B,\ci_B,\ce_B)$ be a
  combinatory pre-model. For any homomorphism $f \colon A \to B$ and
  a sequence $(b_{n})_{n \geq 1}$ of elements of
  $\underline{B}$, there exists a unique homomorphism $\overline{f}
  \colon A[\vX] \to B$ such that $\overline{f} \circ \sigma_{A} = f$
  and $\overline{f}(\vx_n) = b_n$ for each $n \geq 1$. 
\end{lemma}

\begin{remark}
  \label{rem:FunctorT}
   By Lemma \ref{prop:UniversalPolyAlg},  the construction $A \mapsto
   A[\vx]$ determines a functor $T$ on the category of combinatory
   pre-models. The functor $T$ sends each homomorphism $f \colon A \to
   B$ to the unique homomorphism $T(f) \colon A[\vx] \to B[\vx]$ such
   that $T(f)(\vx) = \vx$ and $T(f) \circ \eta_{A} = \eta_{B} \circ
   f$.  By induction on $n \in \Nat$, one can show that 
   %
   % Referee 1: comment 1
   %
   % Clarify the notation T^n A 
   $T^{n}A \;(= \left(\cdots(A[\vx])\cdots)[\vx]\right)$ 
   is
   isomorphic to $A[\vx_1,\dots,\vx_n]$ via the unique homomorphism
  \begin{equation}
    \label{eq:IsoPolynomialTA}
    h_{n} \colon A[\vx_1,\dots,\vx_n]\to T^{n}A 
  \end{equation}
  such that $h_{n} \circ \eta_{n} = \eta_{T^{n-1} A} \circ
  \dots \circ \eta_{A}$ and $h_{n}(\vx_i) = (\eta_{T^{n-1} A} \circ \dots \circ
  \eta_{T^{i}A})(\vx)$ for each $i \leq n$.%
  \footnote{Here $h_n(\vx_n) = \vx$. We also define  $T^0 A = A$.}
\end{remark}

\begin{notation}
  \label{rem:ClosedTerm}
  As usual, the \emph{interpretation} of a term $t \in
  \mathcal{T}(\vX+\underline{A})$ in
  $A$ under a valuation $\rho \colon \Nat \to
  \underline{A}$ is the unique homomorphism $\llbracket \cdot
  \rrbracket_{\rho} \colon \mathcal{T}(\vX+\underline{A}) \to A$
  that extends the identity
  function $\id_{A} \colon A \to A$ with
  respect to the sequence $( \rho(n))_{n \in \Nat}$.
  When $t$ is a closed term, the interpretation $\llbracket t
  \rrbracket_{\rho}$ does not depend on $\rho$, and thus can be written as
  $\llbracket t \rrbracket$. 
  Because of this, we often identify a closed term
  $t$ with its interpretation $\llbracket t \rrbracket \in
  \underline{A}$ and treat $t$ as if it is an element of
  $\underline{A}$.%
  \footnote{For example, term $(\ce, ((\cs, \ck),
  \ci))$ will be identified with $\ce(\cs \ck \ci)$.}
  It should be clear from the context whether closed terms are treated
  as terms (or elements of $A[\vX]$) or elements of
  $\underline{A}$ via the interpretation. When
  a closed term $t$ is treated as an element of $A[\vX]$, however,
  this distinction is irrelevant since we have $t \approx_\vX
  \llbracket t \rrbracket$.  The similar notational
  convention applies to closed terms of 
  $\mathcal{T}(\left\{ \vx_1,\dots,\vx_n \right\}+\underline{A})$.
\end{notation}

Next, recall that an object $X$ of a category is a \emph{retract} of another
object $Y$ if there exist morphisms $s \colon X \to Y$ and $r \colon Y
\to X$ such that $r \circ s = \id_{X}$.
In the context of the combinatory pre-model
$A$,
we have the following.
\begin{proposition}
  \label{prop:Retract}
  \leavevmode
   \begin{enumerate}
     \item\label{prop:Retract2} $A[\vx_1,\dots,\vx_n]$ is a retract of $A[\vX]$
       for each $n \in \Nat$.
     \item\label{prop:Retract3} $(A[\vX])[\vx]$ is isomorphic to $A[\vX]$.
   \end{enumerate}
\end{proposition}
\begin{proof}
  \noindent\ref{prop:Retract2}.
  Fix $n \in \Nat$.
  By Lemma \ref{prop:UnivPropPolynomialAlgX}, 
  there exists a unique homomorphism
  $f \colon A[\vX] \to A[\vx_1,\dots,\vx_n]$ such that $f \circ
  \sigma_{A} =
  \eta_{n}$ and $f(\vx_{i}) = \vx_{\min\{n,i\}}$ for each $i \geq 1$.%
  \footnote{When $n = 0$, we define $\vx_{\min\{n,i\}} = \ci$.}
  In the other direction, by Lemma 
\ref{prop:UniversalPolyAlg}, there is a unique
  homomorphism
  $\sigma_{n} \colon A[\vx_1,\dots,\vx_n]
  \to A[\vX]$ such that $\sigma_{n} \circ
  \eta_{n} = \sigma_{A}$ and $\sigma_{n}(\vx_i) = \vx_{i}$ for $i
  \leq n$.
  Then, $f \circ \sigma_{n} \circ \eta_{n} = f \circ \sigma_{A} =
  \eta_{n}$ and $(f \circ \sigma_{n})(\vx_{i}) = f(\vx_{i}) =
  \vx_{i}$ for each $i \leq n$. Hence $f \circ \sigma_{n} =
  \id_{A[\vx_1,\dots,\vx_n]}$ by Lemma 
  \ref{prop:UniversalPolyAlg}.
  \smallskip

  \noindent\ref{prop:Retract3}.
  By Lemma \ref{prop:UnivPropPolynomialAlgX},
  there exists a unique homomorphism
  $f \colon A[\vX] \to (A[\vX])[\vx]$ such that 
  $f \circ \sigma_{A} = \eta_{A[\vX]} \circ \sigma_{A}$,
  $f(\vx_{1}) = \vx$, and $f(\vx_{i}) = \eta_{A[\vX]}(\vx_{i-1})$ for each 
  $i \geq 2$. In the other direction, there exists a unique
  homomorphism $g \colon A[\vX] \to
  A[\vX]$ such that $g \circ
  \sigma_{A} = \sigma_{A}$ and $g(\vx_{i}) = \vx_{i+1}$
  for each $i \geq 1$. By Lemma \ref{prop:UniversalPolyAlg},
  $g$ extends uniquely to a homomorphism $h \colon (A[\vX])[\vx] \to
  A[\vX]$ such that $h(\vx) = \vx_{1}$ and $h \circ
  \eta_{A[\vX]} = g$. Then, it is straightforward to show that $f$ and
  $h$ are mutual inverse.
\end{proof}
As a corollary of Proposition~\ref{prop:Retract}\eqref{prop:Retract2},
we have $t \approx_{\vx_1\dots\vx_n} u \iff 
t \approx_{\vX} u$ for each $t, u \in \mathcal{T}(\left\{ \vx_1,\dots,
\vx_n\right\} + \underline{A})$.

Before proceeding further, we recall one of the standard abstraction mechanisms for combinatory
algebras (cf.\ Barendregt~\cite[7.3.4]{BarendregtLambda}).
   %
   % Referee 1: comment 2
   %
   % Here, we use Barengregt's book 7.3.4, which is called \lambda^1.
   % The referee ask to use \lambda^\ast defined in 7.1.5.
   % Is is relevant? It can be shown to be irrelevant (You can use
   % 7.1.5, but it would make some of the proof more involved; in
   % particular Lemma 6.18 and Lemma 6.19.)
   %
   % Def 7.3.4 is easy to work with as it is defined on the generators
   % and is freely extended to all polynomials. 
   %
\begin{definition}
  \label{def:LambdaAst}
  For each $t \in \mathcal{T}(\vX + \underline{A})$ and $\vx \in \vX$, define 
  $\lambda^\ast \vx.t \in  \mathcal{T}(\vX + \underline{A})$ inductively by
\begin{enumerate}
\item $\lambda^\ast \vx.\vx = \ci$,
\item $\lambda^\ast \vx.a = \ck a$,
\item $\lambda^\ast \vx.(t,u) = \cs (\lambda^\ast \vx.t) (\lambda^\ast
  \vx.u)$,
\end{enumerate}
where $a \in \vX + \underline{A}$ such that $a \neq \vx$.
For each $n \geq 1$, we write $\lambda^{\ast}\vx_1\dots\vx_n. t$ for
$\lambda^{\ast}\vx_1. \cdots\lambda^{\ast}\vx_n. t$.
\end{definition}
\begin{lemma}
  \label{lem:Beta}
  For each $t,u \in \mathcal{T}(\vX+\underline{A})$
   and $\vx \in \vX$, we have
  $(\lambda^\ast \vx.t) u \approx_\vX t[\vx/u]$, where 
  $t[\vx/u]$ denotes the substitution of $u$ for $\vx$ in $t$.
\end{lemma}
\begin{proof}
  By induction on the complexity of $t$.
\end{proof}
   %
   % Referee 1: comment 3
   %
   % In Notation 2.10, we decided to use x,y,z for x_1, x_2, and x_3.
   % We do not think that any clarity is lost here by following our
   % convention.
   %
\begin{proposition}
  \label{prop:Retract1}
  $A[\vx_1,\dots,\vx_n]$ is a retract of $A[\vx]$ for each $n \in \Nat$.
\end{proposition}
\begin{proof}
  By Lemma \ref{prop:UniversalPolyAlg}, the identity
  $\id_{A} \colon A \to A$ extends to a homomorphism $\overline{\id_{A}}
  \colon A[\vx] \to A$ such that $\overline{\id_{A}}(\vx) = \ci$ and
  $\overline{\id_{A}} \circ \eta_{1} = \id_{A}$.
  Thus, $A$ is a retract of $A[\vx]$. By the functoriality of $T$,
  it remains to show that $A[\vx,\vy] \;(\cong T^{2} A)$ is a retract of
  $A[\vx] \;(= T A)$.
  To see this, define the following terms
  (cf.\ Barendregt~\cite[Section 6.2]{BarendregtLambda}):
  \begin{align}
    \label{eq:Pairing}
    \ct &=  \ck, &
    \cf &= \lambda^\ast \vx\vy. \vy, &
    [\cdot,\cdot] &=\lambda^\ast \vx\vy\vz. \vz \vx \vy.
  \end{align}
  By Lemma \ref{prop:UniversalPolyAlg},
  there exists a unique homomorphism $f \colon A[\vx,\vy] \to A[\vx]$ such that
  $f \circ \eta_{2} = \eta_{1}$, $f(\vx) \approx_\vx \vx \ct$,
  and  $f(\vy) \approx_\vx \vx \cf$. 
  In the other direction, there exists a unique homomorphism $g \colon A[\vx]
  \to A[\vx,\vy]$ such that 
  $g \circ \eta_{1} = \eta_{2}$ and 
  $g(\vx) \approx_{\vx\vy} [\vx,\vy]$.
  Since $g(f(\vx)) \approx_{\vx\vy} g(\vx \ct) 
  \approx_{\vx\vy} [\vx,\vy] \ct  
  \approx_{\vx\vy} \vx$, $g(f(\vy)) \approx_{\vx\vy} g( \vx  \cf) 
  \approx_{\vx\vy} [\vx,\vy]  \cf  \approx_{\vx\vy} \vy$,
  and $g \circ f \circ \eta_{2} = \eta_{2}$,
  we must have $g \circ f = \id_{A[\vx,\vy]}$ by Lemma
  \ref{prop:UniversalPolyAlg}.
\end{proof}
%%%%%%%%%%%%%%%%%%%%%%%%%%%%%%%%%%%%%%%%%%%%%%%%%
The rest of the section concerns an alternative representation
of $A[\vx]$ without using indeterminates.
We begin with the following observation:
since $(\lambda^{\ast}\vx.t) \vx \approx_\vx t$ for each
$t \in \mathcal{T}(\left\{ \vx \right\} + \underline{A})$, a function $f \colon
\underline{A} \to \underline{A[\vx]}$ defined by 
\[
  f(a) = a\vx
\]
is surjective. 
Let $\sim_{A}$ be the equivalence relation on
$\underline{A}$ generated by the kernel 
\[
  \left\{ (a,b) \in \underline{A} \times \underline{A} \mid f(a)
= f(b) \right\}
\]
of $f$, and let $\pi_{\sim_A} \colon \underline{A} \to A/{\sim_{A}}$
be the natural map onto the (set theoretic) quotient of $A$ by $\sim_A$. Then, $f$
uniquely extends
to a bijection $\gamma \colon A/{\sim_{A}} \to \underline{A[\vx]}$
such that $\gamma \circ \pi_{\sim_A}  = f$ with an inverse $\lambda
\colon \underline{A[\vx]} \to A /{\sim_{A}}$ defined by $\lambda(t) =
\lambda^{\ast} \vx. t$.%
\footnote{More precisely, $\lambda(t) =
[\llbracket \lambda^{\ast} \vx. t \rrbracket]_{\sim_{A}}$.}
Thus, $A[\vx]$ induces the following combinatory pre-model structure 
on $A /{\sim_{A}}$:
\[
 \widetilde{A} =
 (\underline{A}/{\sim_{A}}, \qa, \ck\ck, \ck\cs, \ck\ci, \ck\ce)
\]
 where
\(
 a \qa b = \cs a b
\)
(note that $\cs a b \vx \approx_\vx a \vx(b \vx)$).
The point of the following is that 
the relation ${\sim_{A}}$ can be characterised
directly without passing through $A[\vx]$.

\begin{definition}\label{def:quotient1}
Define an applicative structure $A_1 = (\underline{A}, \cdot_1)$ on
$\underline{A}$ by 
\[
  a \cdot_{1} b = \cs ab.
\]
For each $a \in \underline{A}$, define $a_1 \in \underline{A}$ 
by $a_{1} = \ck a$.

Let $\sim_{1}$ be the congruence relation on $A_{1}$ generated by 
the following relation:
\begin{enumerate}
\item\label{eq:quoK} $\ck_{1} \cdot_{1} a \cdot_{1} b \sim_{1} a$,
\item\label{eq:quoS} $\cs_{1} \cdot_{1} a \cdot_{1} b \cdot_{1} c \sim_{1}
  a \cdot_{1} c \cdot_{1} (b \cdot_{1} c)$,
\item\label{eq:quoI} $\ci_{1} \cdot_{1} a \sim_{1} a$,
\item\label{eq:quoE} $\ce_{1} \cdot_{1} a \cdot_{1} b \sim_{1} a \cdot_{1}
  b$,
\item\label{eq:quoInj} $(\ck a) \cdot_{1} (\ck b) \sim_{1}
  \ck  (a b)$,
\item\label{eq:quoEta} $(\ck a) \cdot_{1} \ci \sim_{1} a$,
\end{enumerate}
where  $a, b, c \in \underline{A}$.
Then, the structure
\[
 \bar{A}_1 =
 (\underline{A}/{\sim_1}, \qna_1, \ck_1, \cs_1, \ci_1,
 \ce_1),
\]
where $(\underline{A}/{\sim_{1}}, \ast_1)$ is the quotient of $A_{1}$
by
$\sim_1$, is a combinatory pre-model.
\end{definition}
The crucial axiom
of $\sim_1$ is Definition \ref{def:quotient1}\eqref{eq:quoEta},
as can be seen from the proof of the following
theorem.
\begin{theorem}
  \label{thm:IsoPolylBar}
  Combinatory pre-models $\bar{A}_1$ and $A[\vx]$ are isomorphic.
\end{theorem}
\begin{proof}
  First, a function $f \colon \mathcal{T}(\{\vx\}+\underline{A}) \to
  \underline{A}/{\sim_1}$ defined by $f(t) = \lambda^{\ast} \vx. t$
  is a homomorphism of applicative structures $\mathcal{T}(\{\vx\}+\underline{A})$
  and $\bar{A}_1$. It is also straightforward to check that 
  $f$ preserves \eqref{eq:polyK}--\eqref{eq:polyInj} of
  Definition~\ref{def:PolyAlg}.
  Hence, by Proposition~\ref{prop:AppUniExt}, $f$ extends uniquely to
  a homomorphism $\overline{f} \colon A[\vx] \to \bar{A}_1$ of the
  underlying applicative structures.
  Then, it is easy to see that $\overline{f}$ is a homomorphism of
  combinatory pre-models.
  In the other direction, a function $g \colon \underline{A} \to
  \underline{A[\vx]}$ defined by $g(a) = a \vx$ is a homomorphism of
  applicative structures $A_{1}$ and $A[\vx]$, and it is easy to see
  that $g$ preserves \eqref{eq:quoK}--\eqref{eq:quoEta} of
  Definition~\ref{def:quotient1}. By a similar argument as above,  $g$ extends uniquely to a
  homomorphism $\overline{g} \colon \bar{A}_1 \to A[\vx]$.   
  Lastly,
  we have
  $(\overline{g} \circ f) (t) \approx_\vx  ( \lambda^{\ast} \vx. t) \vx \approx_\vx t$ for each $t \in \mathcal{T}(\{\vx\}+\underline{A})$
  and
  $(\overline{f} \circ g) (a)
  \sim_1 \lambda^{\ast} \vx. a  \vx
  \sim_1 \ck a \cdot_1 \ci
  \sim_1 a$ 
  for each $a \in \underline{A}$, where the last equality follows from
  the equation \eqref{eq:quoEta}.
  Thus $\overline{f}$ and $\overline{g}$ are
  inverse to each other. 
\end{proof}

From the proof of Theorem \ref{thm:IsoPolylBar}, we can derive the following correspondence.
\begin{proposition}\label{cor:CongrCorr}
  For any combinatory pre-model $A$, we have
  \[
    a \sim_1 b \iff a \vx \approx_{\vx} b \vx
  \]
  for each $a,b \in \underline{A}$.
\end{proposition}
As an immediate consequence, we have the following.
\begin{corollary}\label{cor:EinSim}
  For any combinatory pre-model $A$, we have $\ce a \sim_1 a$ for each
  $a \in \underline{A}$.
\end{corollary}

%%%%%%%%%%%%%%%%%%%%%%%%%%%%%%%%%%%%%%%%%%%%%%%%
\section{Reflexivity}\label{sec:Reflexivity}
%%%%%%%%%%%%%%%%%%%%%%%%%%%%%%%%%%%%%%%%%%%%%%%%
The notion of reflexivity introduced below can be understood as
an algebraic (rather than first-order) analogue of the Meyer--Scott
axiom for combinatory models (cf.\ Definition~\ref{def:CombMod}).
We still follow the convention of Notation~\ref{not:CPM}.
\begin{definition}
  \label{defreflexive}
  A combinatory pre-model $A$ is \emph{reflexive} if 
  \[
    a \sim_{1} b \implies \ce a = \ce b
  \]
  for each $a, b \in \underline{A}$.
\end{definition}
In terms of polynomial algebras, reflexivity can be stated as follows
(cf.\ Proposition \ref{cor:CongrCorr}). 
\begin{lemma}
  \label{lem:reflexive}
  A  combinatory pre-model $A$ is reflexive
  if and only if
  $a \vx \approx_\vx b \vx$ implies $\ce a = \ce b$ for each $a, b \in
  \underline{A}$.
\end{lemma}
In the rest of the paper, we sometimes use Lemma \ref{lem:reflexive}
implicitly.

Since the relation $\sim_{1}$ is generated from the equations on the
elements on $\underline{A}$ (cf.\ Definition \ref{def:quotient1}),
reflexivity can be characterised by a set of simple universal
sentences on $A$.
\begin{proposition}
  \label{prop:CharRef}
  A combinatory pre-model $A$ is reflexive if and only if
  it satisfies the following equations:
  \begin{enumerate}
  \item\label{prop:CharRef1}
    $\ce(\ck_{1} \cdot_{1} a \cdot_{1} b)
    =
    \ce a$,
  \item\label{prop:CharRef2}
    $\ce(\cs_{1} \cdot_{1} a \cdot_{1} b \cdot_{1} c)
    = 
    \ce(a \cdot_{1} c \cdot_{1} (b \cdot_{1} c))$,
  \item\label{prop:CharRef3}
    $\ce(\ci_{1} \cdot_{1} a)
    =
    \ce a$,
  \item\label{prop:CharRef4}
    $\ce(\ce_{1} \cdot_{1} a \cdot_{1} b)
    =
    \ce(a \cdot_{1} b)$,
  \item\label{prop:CharRef5}
    $\ce((\ck a) \cdot_{1} (\ck b))
    =
    \ce(\ck (a b))$,
  \item\label{prop:CharRef6}
    $\ce((\ck a) \cdot_{1} \ci)
    =
    \ce a$,
  \item\label{prop:CharRef7}
    $\ce((\ce a) \cdot_{1} (\ce b))
    =
    \ce (a \cdot_{1} b)$
  \end{enumerate}
  for each $a, b, c \in \underline{A}$.
\end{proposition}
\begin{proof}
  \noindent ($\Rightarrow$) 
  Suppose that $A$ is reflexive. Then
  \eqref{prop:CharRef1}--\eqref{prop:CharRef6}
  hold by Definition \ref{def:quotient1}. Moreover, we have
  $\ce a \cdot_{1} \ce b \sim_{1} a \cdot_{1} b$ by 
  Corollary~\ref{cor:EinSim}, from which \eqref{prop:CharRef7} follows.

 \noindent($\Leftarrow$)
  Suppose that $A$ satisfies
  \eqref{prop:CharRef1}--\eqref{prop:CharRef7} above. Let $A_{1} =
  (\underline{A}, \cdot_{1})$ be the applicative structure defined
  as in Definition \ref{def:quotient1}, and define another applicative structure
  $A_{\ce} = (A^{\ast}, \cdot_{\ce})$ on $A^{\ast} = \left\{ \ce a
    \mid a \in \underline{A} \right\}$ by
    $
    \ce a \cdot_{\ce} \ce b = \ce( (\ce a) \cdot_{1} (\ce b)).
    $
  Then, a function $f \colon \underline{A} \to
  \underline{A}$ given by  $f(a) = \ce a$ is a homomorphism from
  $A_{1}$ to $A_{\ce}$ by \eqref{prop:CharRef7}, which moreover
  preserves relations \eqref{eq:quoK}--\eqref{eq:quoEta} of Definition
  \ref{def:quotient1} by 
  \eqref{prop:CharRef1}--\eqref{prop:CharRef6} above. Thus,  
  $f$ extends uniquely to a homomorphism 
  from $\bar{A}_1$ to $A_{\ce}$. In particular, $a \sim_{1} b$
  implies $\ce a = \ce b$. Hence $A$ is reflexive.
\end{proof}
\begin{remark}
  The condition \eqref{prop:CharRef7} in Proposition \ref{prop:CharRef}
  does not have a counterpart in Definition \ref{def:quotient1}, but
  it seems to be needed here. This is because $\ce$ is required to preserve 
  the congruence relation generated by \eqref{eq:quoK}--\eqref{eq:quoEta}
  of Definition \ref{def:quotient1}, and not the (weaker) equivalence
  relation.
\end{remark}

Since the seven equations of Proposition \ref{prop:CharRef} are 
simple universal sentences, we have the following corollary.  Here, a
\emph{substructure} of a combinatory pre-model $A$ is a subset $B
\subseteq \underline{A}$ which contains $\ck$, $\cs$, $\ci$, $\ce$,
and is closed under the application of $A$.
\begin{corollary}
  \label{prop:ReflexiveRetract}
  Reflexivity is closed under substructures and homomorphic images.
  In particular, it is closed under retracts.
\end{corollary}

As a consequence, reflexivity is preserved under the addition of further
indeterminates to a  polynomial algebra.
\begin{proposition}
  \label{lem:N1reflexNreflex}
    For any combinatory pre-model $A$, 
    if $A[\vx]$ is reflexive, then $A[\vx_1,\dots,\vx_n]$ is reflexive for each $n \in
   \Nat$.
\end{proposition}
\begin{proof}
  Immediate from Corollary \ref{prop:ReflexiveRetract} and
  Proposition \ref{prop:Retract1}.
\end{proof}
Note that we do not necessarily have that $A[\vx]$ is reflexive when
$A$ is. 

Next, we introduce an alternative abstraction mechanism for
combinatory pre-models. This abstraction mechanism allows us to see reflexivity
as a requirement that the equality of $A[\vx]$ be preserved by the
abstraction.
\begin{definition}
  \label{def:LambdaDag}
  Let $A$ be a combinatory pre-model.
  For each $t \in \mathcal{T}(\vX+\underline{A})$ and $\vx \in \vX$,
  define $\lambda^\dag \vx.t \in \mathcal{T}(\vX+\underline{A})$ inductively by
\begin{enumerate}
\item\label{def:LambdaDag1}
 $\lambda^\dag \vx.\vx = \ce\ci$,
\item\label{def:LambdaDag2}
 $\lambda^\dag \vx.a = \ce (\ck a)$,
\item\label{def:LambdaDag3}
 $\lambda^\dag \vx.(a, \vx) = \ce a$,
\item\label{def:LambdaDag4}
 $\lambda^\dag \vx.(t,u) = \ce (\cs (\lambda^\dag \vx.t)
 (\lambda^\dag \vx.u))$ otherwise,
\end{enumerate}
where $a \in \vX + \underline{A}$ such that $a \neq \vx$.
As in Lemma~\ref{lem:Beta}, we have $(\lambda^\dag \vx. t) u
\approx_{\vX} t[\vx/u]$ for each 
$t,u \in \mathcal{T}(\vX+\underline{A})$.
\end{definition}
First, we note that the choice of
abstraction mechanisms does
not affect the isomorphism between $A[\vx]$ and $\bar{A}_1$ in Theorem
\ref{thm:IsoPolylBar}.
\begin{proposition}
  \label{prop:LambdaDagLambdaAst}
  For any combinatory pre-model $A$, 
  we have $\lambda^\ast \vx. t \sim_{1} \lambda^\dag \vx. t$
  for each $t \in \mathcal{T}(\{\vx\}+\underline{A})$.
\end{proposition}
\begin{proof}
  Immediate from Proposition~\ref{cor:CongrCorr}.
\end{proof}

Nevertheless, $\lambda^\dag$-abstraction enjoys some properties
which $\lambda^\ast$-abstraction need not. The following is crucial
for our development.
\begin{lemma}
  \label{lem:ReflE}
  If $A$ is a reflexive combinatory pre-model, then $\ce (\ce a) = \ce a$ for each $a \in
  \underline{A}$. In particular, we have $\ce (\lambda^{\dag}\vx.t) =
  \lambda^{\dag}\vx.t$ for each $t \in \mathcal{T}(\{\vx\}+\underline{A})$.
\end{lemma}
\begin{proof}
  By Corollary \ref{cor:EinSim} and the definition of reflexivity. 
\end{proof}
We can now characterise reflexivity in terms of $\lambda^\dag$-abstraction.
The proposition below says that reflexivity amounts to the
requirement that the mapping $t \mapsto \lambda^\dag \vx. t$ be
well-defined on the polynomials with one indeterminate.
\begin{proposition}
  \label{lem:ReflexLambda}
  A combinatory pre-model $A$ is reflexive if and only if 
  \[
    t \approx_\vx u \implies \lambda^\dag \vx. t = \lambda^\dag \vx. u
  \]
  for each $t, u \in \mathcal{T}(\{\vx\}+\underline{A})$.
\end{proposition}
\begin{proof}
  Suppose that $A$ is reflexive, and let $t \approx_{\vx} u$.
  Then 
  $\lambda^\dag \vx. t \sim_1 \lambda^\dag \vx. u$ by Proposition
  \ref{prop:LambdaDagLambdaAst},
  and so $\ce (\lambda^\dag \vx. t) = \ce (\lambda^\dag \vx. u)$
  by reflexivity.
  By  Lemma \ref{lem:ReflE}, we obtain
  $\lambda^\dag \vx. t = \lambda^\dag \vx. u$.
  The converse is immediate from Definition
  \ref{def:LambdaDag}\eqref{def:LambdaDag3}.
\end{proof}

The reflexivity of polynomial algebras admits similar
characterisations.
\begin{corollary}
  \label{cor:ReflexPolynomial}
  For each $n \in \Nat$, the following are equivalent:
  \begin{enumerate}
    \item $A[\vx_1,\dots,\vx_n]$ is reflexive.

    \item \label{lem:ReflexPolynomial2}
    $t \approx_{\vx_1\dots\vx_{n+1}} u 
    \implies
    \lambda^\dag \vx_{n+1}. t \approx_{\vx_1\dots\vx_n} \lambda^\dag
    \vx_{n+1}. u$
  for each $t, u \in \mathcal{T}(\left\{ \vx_1,\dots, \vx_{n+1} \right\} + \underline{A})$.

    \item \label{lem:ReflexPolynomial3}
      $t \vx_{n+1} \approx_{\vx_1\dots\vx_{n+1}} u  \vx_{n+1}
    \implies
    \ce t \approx_{\vx_1\dots\vx_n} \ce u$
  for each $t, u \in \mathcal{T}(\left\{ \vx_1,\dots, \vx_{n} \right\} + \underline{A})$.
  \end{enumerate}
\end{corollary}
\begin{proof}
  By Remark \ref{rem:FunctorT}, it suffices to show that items
  \ref{lem:ReflexPolynomial2}
  and 
  \ref{lem:ReflexPolynomial3}
    are equivalent to the reflexivity of $T^n A$.
    To this end, consider the following commutative diagram:
  \[
    \xymatrix@C-2ex@M+2pt{
      T^{n+1}A
      & A[\vx_1,\dots,\vx_{n+1}] \ar [l]^-{\cong}_-{h_{n+1}} 
      & \\
      \mathcal{T}(\left\{ \vx \right\} + \underline{T^n A})
      \ar@{>>} [u]^(.4){\pi}
      & 
      &  \mathcal{T}(\left\{ \vx_1,\dots,\vx_{n+1} \right\} +
      \underline{A}) \ar@{>>} [ll]_(.4){\psi} 
      \ar@{>>} @(u,l) [ul]_-{\pi} 
      \\
      T^n A \ar@{^{(}->} [u]
      \ar  @/^40pt/  [uu]^{\eta_{T^{n}A}}
      & A[\vx_1,\dots,\vx_{n}] \ar [l]^{\cong}_{h_{n}} \ar@{^{(}->} [uu] |\hole
      & \\
      A \ar@{^{(}->} [r] \ar [u]^{\eta_{T^{n-1}A} \circ  \dots \circ \eta_{A}}
      & 
      \mathcal{T}(\left\{ \vx_1,\dots,\vx_n \right\} + \underline{A})
      \ar@{>>} [u]_(.4){\pi} 
      \ar@{>>} [ul]^-{\varphi}
      \ar@{^{(}->} @(dr,d) [uur]
      & 
    }
  \]
  Here 
  \begin{itemize}
    \item
    $h_{k} \colon A[\vx_1,\dots,\vx_k] \to T^{k}
    A$ is the isomorphism of \eqref{eq:IsoPolynomialTA} for each $k
    \in \Nat$;

    \item 
      $\varphi \colon \mathcal{T}(\left\{ \vx_1,\dots,\vx_n \right\}
      + \underline{A}) \to T^{n}A$ is the unique extension of
      $\eta_{T^{n-1}A} \circ \dots \circ \eta_{A}$ such that $\varphi(\vx_i) = 
      ( \eta_{T^{n-1}A} \circ \dots \circ \eta_{T^iA} )(\vx)$ for each $i \leq n$;

    \item
      $\psi \colon \mathcal{T}(\left\{ \vx_1,\dots,\vx_{n+1} \right\}
      + \underline{A}) \to \mathcal{T}(\left\{ \vx \right\} +
      \underline{T^{n}A})$ is the unique extension of $\varphi$ such
      that $\psi(\vx_{n+1}) = \vx$;

    \item each $\pi$ is the quotient map with respect to the
      congruence relation of Definition \ref{def:PolyAlg};

    \item the other unnamed maps are the natural inclusions.
  \end{itemize}
  By induction on the complexity of terms,
  one can show that 
  \begin{equation}
    \label{eq:lem:ReflexPolynomial2}
    \varphi(\lambda^\dag \vx_{n+1}. t) = \lambda^\dag \vx. \psi(t)
  \end{equation}
  for each $t \in \mathcal{T}(\left\{ \vx_1,\dots, \vx_{n+1} \right\}
  + \underline{A})$. % (cf.\ Remark \ref{rem:ClosedTerm}).
  Then, since $\psi$ is surjective,
    $T^n A$ is reflexive if and only if
   %
   % Referee 1: comment 5
   %
   % Typo?: No, this is what we mean as we regard T^{n+1} A = (T^n A) [x].
   %
  \begin{equation}
    \label{eq:lem:ReflexPolynomial3}
    \psi(t) \approx_{\vx} \psi(u) \implies 
    \lambda^{\dag} \vx. \psi(t) = \lambda^{\dag} \vx. \psi(u)
  \end{equation}
  for each $t,u \in \mathcal{T}(\left\{ \vx_1,\dots, \vx_{n+1} \right\}
  + \underline{A})$ by Proposition \ref{lem:ReflexLambda}.
  By equation \eqref{eq:lem:ReflexPolynomial2} and
  the commutativity of 
  the above diagram, \eqref{eq:lem:ReflexPolynomial3} is equivalent to 
  item \eqref{lem:ReflexPolynomial2}.

  By the similar argument using Lemma \ref{lem:reflexive}, $T^n A$ is reflexive if and only if
  \begin{equation*}
    \varphi(t)\vx \approx_{\vx} \varphi(u)\vx \implies 
    \ce \varphi(t) = \ce \varphi(u)
  \end{equation*}
  for each $t,u \in \mathcal{T}(\left\{ \vx_1,\dots, \vx_{n} \right\}
  + \underline{A})$, which is equivalent to 
  \begin{equation}
    \label{eq:lem:ReflexPolynomial4}
    \psi(t \vx_{n+1}) \approx_{\vx} \psi(u\vx_{n+1}) \implies 
    \varphi(\ce t) = \varphi(\ce u)
  \end{equation}
  for each $t,u \in \mathcal{T}(\left\{ \vx_1,\dots, \vx_{n} \right\}
  + \underline{A})$. By the commutativity of 
  the above diagram, \eqref{eq:lem:ReflexPolynomial4} is equivalent to item \eqref{lem:ReflexPolynomial3}.
\end{proof}

To close this section, we show that the reflexivity of $A$ allows us
to represent $A[\vx]$ by a structure on the fixed-points of $\ce$. 
The reader should recall some notations from Definition
\ref{def:quotient1}.
   %
   % Referee 1: comment 6
   %
   % The following makes two construction at once.
   % Is it useful to split these.
   % Maybe, we can move the definition of B_A into Proposition 3.13.
   %
   %
   % Referee 1: comment 7
   %
   % Change B_A to A^{\ast}? Ok in Proposition 3.13.
   %
\begin{definition}
\label{def:B}
  For a reflexive combinatory pre-model $A$, define 
  \[
  A^\ast = \left\{ a \in \underline{A} \mid \ce a = a \right\},
  \]
  or equivalently $A^\ast = \left\{ \ce a \mid a \in \underline{A} \right\}$ (cf.\ Lemma \ref{lem:ReflE}).
\end{definition}

  Now, define a combinatory pre-model structure on $A^\ast$ by
  \[
    A^\ast = (A^\ast, \diamond, \ce\ck_1, \ce\cs_1, \ce\ci_1, \ce\ce_1),
  \]
  where $a \diamond b =  \ce(a\cdot_1 b)$ for each $a,b  \in A^{\ast}$.
\begin{proposition}
  \label{prop:BEquivAx}
  If $A$ is a reflexive combinatory pre-model, then 
  $A^\ast \simeq \bar{A}_1$, and hence $A^\ast \simeq A[\vx]$.
\end{proposition}
\begin{proof}
  Let $f \colon A \to A^\ast$ be a function defined by $f(a) = \ce a$.
  As the proof of Proposition \ref{prop:CharRef} shows, 
  $f$ extends uniquely to an isomorphism $\bar{f} \colon \bar{A}_1 \to A^\ast$
  of applicative structures. By the very definition of
  $\bar{A}_1$ and $A^\ast$, $\bar{f}$ is a homomorphism of
  combinatory pre-models. 
\end{proof}

\begin{remark}
  \label{rem:Krivine}
Constructions similar to $A^\ast$ have appeared in the
literature~\cite{KrivineLambda,FreydCombinator,Selinger02}.  For
example, Krivine \cite[Section 6.3]{KrivineLambda} defined an applicative
structure $B = (A^*, \qa, \ck\ck, \ck\cs)$ with $a \qa b = \cs a b$
for a combinatory algebra $A = (\underline{A}, \cdot, \ck,\cs)$
under the slightly stronger condition than reflexivity.
His condition consists of \eqref{prop:CharRef1},
\eqref{prop:CharRef2}, and 
\eqref{prop:CharRef5} of Proposition \ref{prop:CharRef} without $\ce$
in front of both sides of the equations,
together with the following weak form of stability:%
\footnote{Krivine defined $\ce$ by $\ce = \lambda^\ast \vx\vy.\vx\vy$,
which makes \eqref{prop:CharRef6} of Proposition \ref{prop:CharRef}
(without $\ce$ in front of the left side) superfluous. Moreover,
under Krivine's condition, the equation \eqref{prop:CharRef7}
is derivable from the other equations. In short, Krivine's condition
consists of instantiations of three of Curry's axioms for lambda
algebras (cf.\ Definition \ref{def:LambdaAlg}) together with \eqref{eq:WStability}.
}
\begin{equation}
  \label{eq:WStability}
  \forall a,b \in \underline{A} \left[ \ce(\ck a) = \ck a \mathrel{\&} \ce(\cs a b) =  \cs a
  b \right].
\end{equation}
It is clear that the structures $A^\ast$ and $B$ coincide in Krivine's
context when one ignores $\ci$ and $\ce$; in this case, $B$ is
isomorphic to $A[\vx]$. The same representation of $A[\vx]$ as $B
$ can be found in Selinger~\cite[Proposition~4]{Selinger02} in the
case where $A$ is a lambda algebra.
In view of this, Proposition~\ref{prop:BEquivAx} generalises
the construction  of the previous works.
\end{remark}

%%%%%%%%%%%%%%%%%%%%%%%%%%%%%%%%%%%%%%%%%
\section{Strong reflexivity}\label{sec:AlgCombModel}
%%%%%%%%%%%%%%%%%%%%%%%%%%%%%%%%%%%%%%%%%
In this section, we introduce a stronger notion of reflexivity, which
can be seen as an algebraic analogue of the Meyer--Scott axiom
for combinatory models.
In the following, the conventions of Notation \ref{not:CPM} and
Notation~\ref{rem:ClosedTerm} still apply.
\begin{definition}
  \label{def:AlgCombMod}
  A combinatory pre-model 
  $A = (\underline{A},\mathbin{\cdot},\ck,\cs,\ci,\ce)$ 
  is \emph{strongly reflexive} if $A[\vx]$ is reflexive.
\end{definition}
We first note the following.
\begin{lemma}
  \label{lem:AlgCombMod}
  A combinatory pre-model $A$ is strongly reflexive if and only if
  $A[\vx_1,\dots,\vx_n]$ is reflexive for each $n \in \Nat$.
\end{lemma}
\begin{proof}
  By Proposition \ref{lem:N1reflexNreflex}.
\end{proof}
In particular, strong reflexivity implies reflexivity.
From the above lemma, we obtain the following characterisation. 
\begin{proposition}
  \label{cor:AlgCombModelLamdaDag}
  The following are equivalent for a combinatory pre-model $A$:
   \begin{enumerate}
     \item \label{eq:AlgCombModelLamdaDag1} $A$ is strongly reflexive.

     \item \label{eq:AlgCombModelLamdaDag2}
     $
     t \vx \approx_{\vX} u \vx
     \implies
     \ce t \approx_{\vX} \ce u
     $
   for each $\vx \notin \FV(tu)$ and $t,u \in \mathcal{T}(\vX + \underline{A})$.

     \item \label{eq:AlgCombModelLamdaDag3}
     $
     t \approx_{\vX} u
     \implies
     \lambda^\dag \vx. t \approx_{\vX} \lambda^\dag \vx. u
     $
   for each $\vx \in \vX$ and $t,u \in \mathcal{T}(\vX +
   \underline{A})$.
   \end{enumerate}
\end{proposition}
\begin{proof}
  \noindent ($\ref{eq:AlgCombModelLamdaDag1} \leftrightarrow
  \ref{eq:AlgCombModelLamdaDag2}$)
  By Lemma \ref{lem:AlgCombMod},
  $A$ is strongly reflexive if and only if
  Corollary \ref{cor:ReflexPolynomial}\eqref{lem:ReflexPolynomial3}
  holds for each $n \in \Nat$.
  This is equivalent to \eqref{eq:AlgCombModelLamdaDag2} 
  by a suitable rearrangement of indeterminates $\FV(tu) \cup \left\{ \vx \right\}$ using Lemma \ref{prop:UniversalPolyAlg}.

  \noindent ($\ref{eq:AlgCombModelLamdaDag1} \leftrightarrow
  \ref{eq:AlgCombModelLamdaDag3}$) Similar.
\end{proof}

In particular, the polynomial algebra of a strongly reflexive combinatory
pre-model is closed under the $\xi$-rule with respect to the
abstraction mechanism $\lambda^{\dag} \vx$. In this way, each strongly
reflexive combinatory pre-model gives rise to a model of the lambda
calculus.

Next, we show that the class of strongly reflexive combinatory
pre-models is axiomatisable
with a finite set of closed equations which can be obtained by 
taking $\lambda^\dag$-closures of both sides of the equations
\eqref{prop:CharRef1}--\eqref{prop:CharRef7} of Proposition
\ref{prop:CharRef}.
\begin{theorem}
  \label{Thm:AlgCombMod}
  A combinatory pre-model $A$ is strongly reflexive 
  if and only if it satisfies the following equations:
  \begin{enumerate}
    \item\label{def:AlgCombMod1}
      $\lambda^{\dag}\vx\vy. \ce(\cs(\cs (\ck\ck) \vx) \vy)
      =
      \lambda^\dag \vx\vy. \ce \vx$,
    \item\label{def:AlgCombMod2}
      $\lambda^{\dag}\vx\vy\vz. \ce(\cs(\cs(\cs(\ck \cs)\vx)\vy)\vz)
      =
      \lambda^\dag \vx\vy\vz. \ce (\cs(\cs \vx\vz)(\cs\vy\vz))$,
    \item\label{def:AlgCombMod3}
      $\lambda^{\dag}\vx. \ce(\cs(\ck\ci)\vx)
      =
      \lambda^\dag \vx. \ce \vx$,
    \item\label{def:AlgCombMod4}
      $\lambda^{\dag}\vx\vy. \ce(\cs(\cs(\ck\ce) \vx)\vy)
      =
      \lambda^\dag \vx\vy. \ce (\cs\vx\vy)$,
    \item\label{def:AlgCombMod5}
      $\lambda^{\dag}\vx\vy. \ce(\cs (\ck\vx) (\ck\vy))
      =
      \lambda^\dag \vx\vy. \ce (\ck(\vx\vy))$,
    \item\label{def:AlgCombMod6}
      $\lambda^{\dag}\vx. \ce(\cs(\ck\vx)\ci)
      =
      \lambda^\dag \vx. \ce \vx$,
    \item\label{def:AlgCombMod7}
      $\lambda^{\dag}\vx\vy. \ce(\cs(\ce\vx) (\ce\vy))
      =
      \lambda^\dag \vx\vy. \ce (\cs\vx\vy)$.
  \end{enumerate}
\end{theorem}
\begin{proof}
  \noindent ($\Rightarrow$) 
  Suppose that $A$ is strongly reflexive. By
  Lemma~\ref{lem:AlgCombMod},
  $A[\vx_1,\dots,\vx_n]$ is reflexive for each $n \in \Nat$.
  Then, \eqref{def:AlgCombMod1}--\eqref{def:AlgCombMod7} follow from the reflexivity of
  $A[\vx_1,\dots,\vx_n]$ ($n \leq 3$) together with
  Proposition~\ref{prop:CharRef} and
  Corollary~\ref{cor:ReflexPolynomial}\eqref{lem:ReflexPolynomial2}.
  
 \noindent($\Leftarrow$)
  Suppose that $A$ satisfies
  \eqref{def:AlgCombMod1}--\eqref{def:AlgCombMod7}.
  Since these equations consist of constants $\ck,\cs,\ci,\ce$ only,
  $A[\vx]$ satisfies these equations as well. Then,
  $A[\vx]$ is reflexive by Proposition \ref{prop:CharRef}, i.e., 
  $A$ is strongly reflexive.
\end{proof}
Note that the equations
\eqref{def:AlgCombMod1}--\eqref{def:AlgCombMod7} in Theorem~\ref{Thm:AlgCombMod}
are \emph{closed} in
the sense that they correspond to terms built up from $\ck,\cs,\ci,\ce$ only.
Thus, we have the following.
\begin{proposition}
  \label{prop:ImageAlgCombModel}
If $A$ is strongly reflexive and $f \colon A \to B$ is a homomorphism
of combinatory pre-models, then $B$ is strongly reflexive.
\end{proposition}

Next, we relate strongly reflexive combinatory pre-models and combinatory models.
\begin{definition}[Meyer~\cite{MeyerLambda}]
  \label{def:CombMod}
  A \emph{combinatory model}
  is a combinatory pre-model 
  \(
   A
  \)
  satisfying the Meyer--Scott axiom:
  \begin{equation*}
    \left[ \forall c \in \underline{A} \left( ac = bc \right) \right] \implies {\ce a = \ce b}
  \end{equation*}
  for each $a, b \in \underline{A}$.
\end{definition}
Note that every combinatory model is reflexive by Lemma~\ref{lem:reflexive}.  
\begin{lemma}
  \label{lem:CombModAlg}
  Every combinatory model is strongly reflexive.
\end{lemma}
\begin{proof}
  Let $A$ be a combinatory model.
  It suffices to show that Corollary
  \ref{cor:ReflexPolynomial}\eqref{lem:ReflexPolynomial3} holds for $n
  =1$.
  Let $t,u \in \mathcal{T}(\left\{ \vx \right\} + \underline{A})$
  be such that $t \vy \approx_{\vx\vy} u \vy$.
  Fix $c
  \in \underline{A}$. Since $(\lambda^\dag \vx. t)\vx \vy
  \approx_{\vx\vy} (\lambda^\dag \vx. u) \vx \vy$, we have $(\lambda^\dag \vx. t)c \vx
  \approx_{\vx} (\lambda^\dag \vx. u) c \vx$.  This implies
  $\ce ((\lambda^\dag \vx. t) c) = \ce((\lambda^\dag \vx. u) c)$
  by the reflexivity of $A$, which is equivalent to
  $(\lambda^\dag \vx.\ce ((\lambda^\dag \vx. t) \vx))c = 
  (\lambda^\dag \vx.\ce ((\lambda^\dag \vx. u) \vx))c$.
  Since $c$ was arbitrary and $A$ is a combinatory model, we have
  $\ce(\lambda^\dag \vx.\ce ((\lambda^\dag \vx. t) \vx)) =
  \ce(\lambda^\dag \vx.\ce ((\lambda^\dag \vx. u) \vx))$,
  and hence $\ce t \approx_{\vx} \ce u$.
\end{proof}

\begin{theorem}
  \label{thm:CombModel}
  The following are equivalent for a combinatory pre-model $A$:
  \begin{enumerate}
    \item\label{thm:CombModel1} $A$ is strongly reflexive.
    \item\label{thm:CombModel2} $A[\vX]$ is strongly reflexive.
    \item\label{thm:CombModel3} $A[\vX]$ is reflexive.
    \item\label{thm:CombModel4} $A[\vX]$ is a combinatory model.
  \end{enumerate}
\end{theorem}
\begin{proof}
  \noindent ($\ref{thm:CombModel1} \leftrightarrow \ref{thm:CombModel2}$)
    By Proposition \ref{prop:ImageAlgCombModel}
    and Proposition~\ref{prop:Retract}~\eqref{prop:Retract2}.
  \smallskip

  \noindent ($\ref{thm:CombModel2} \leftrightarrow
  \ref{thm:CombModel3}$)
  By Corollary \ref{prop:ReflexiveRetract}
  and Proposition \ref{prop:Retract}~\eqref{prop:Retract3}.
  \smallskip

  \noindent($\ref{thm:CombModel4} \rightarrow \ref{thm:CombModel2}$)
  By Lemma \ref{lem:CombModAlg}.
  
  \noindent($\ref{thm:CombModel1} \rightarrow \ref{thm:CombModel4}$)
  Suppose that $A$ is strongly reflexive.
  Let $t,u \in
  \mathcal{T}(\vX + \underline{A})$, and 
  suppose that $ts \approx_{\vX} us$
  for all $s \in \mathcal{T}(\vX + \underline{A})$.
  Choose $n \in \Nat$ such that $t,u \in \mathcal{T}(\left\{ \vx_1,
    \dots,\vx_n
  \right\} + \underline{A})$. 
  Then $t\vx_{n+1} \approx_{\vx_1\dots\vx_{n+1}} u\vx_{n+1}$, so
  by Lemma \ref{lem:CombModAlg} and Corollary \ref{cor:ReflexPolynomial}\eqref{lem:ReflexPolynomial3},
  we have $\ce t \approx_{\vx_1\dots\vx_n} \ce u$.
  Then $\ce t \approx_{\vX} \ce u$.
\end{proof}

It is known that lambda algebras are exactly the retracts of lambda
models. The following is
its analogue for combinatory models.
\begin{theorem}
  \label{thm:CatAlgCombModels}
  A combinatory pre-model is strongly reflexive if and only if it is a
  retract of a combinatory model.
\end{theorem}
\begin{proof}
  \noindent($\Rightarrow$) By Proposition \ref{prop:Retract}~\eqref{prop:Retract2}
  and Theorem~\ref{thm:CombModel}.

  \noindent($\Leftarrow$) By Lemma~\ref{lem:CombModAlg} and
  Proposition \ref{prop:ImageAlgCombModel}.
\end{proof}

%%%%%%%%%%%%%%%%%%%%%%%%%%%%%%%%%%%%%%%%%
\section{Cartesian closed monoids}\label{sec:CCCMoid}
%%%%%%%%%%%%%%%%%%%%%%%%%%%%%%%%%%%%%%%%%
%
In this section, we generalise the construction of a cartesian closed
category with a reflexive object from a lambda algebra due to
Scott~\cite{ScottLambda} (see also
Koymans~\cite{KoymansModelOfLambda}) to
the setting of strongly reflexive combinatory pre-models.  To this
end, we construct a cartesian closed monoid from a strongly reflexive
combinatory pre-model. The connection between cartesian closed
monoids and cartesian closed categories with reflexive objects will be
reviewed toward the end of this section.  The reader is referred to
Koymans \cite[Chapter 2]{KoymansThesis}, Lambek and Scott~\cite[Part
I, 15--17]{LambekScott}, and Hyland~\cite{hyland_2017,HYLAND201459}
for a detailed account of cartesian closed monoids and their relation
to untyped lambda calculus.

Throughout this section, we work over a fixed \emph{reflexive} combinatory pre-model
  \(
   A = (\underline{A},\mathbin{\cdot},\ck,\cs,\ci,\ce).
  \)

First, we construct a monoid out of $A$.
For each $a,b \in \underline{A}$, define
\[
  a \circ b = \lambda^\dag \vx.a (b \vx).
\]
\begin{lemma}\label{monoid}
  For each $a , b, c \in \underline{A}$, we have
  \begin{enumerate}
  \item
    \(
    a \circ (b \circ c) = (a \circ b) \circ c,
    \)\label{monoid1}
  \item
    \(
    \ci \circ a = a \circ \ci = \ce a.
    \)\label{monoid2}
  \end{enumerate}
\end{lemma}
\begin{proof}
  We use Lemma \ref{lem:reflexive} and Lemma \ref{lem:ReflE}.
  \smallskip

  \noindent\ref{monoid1}.
  Since
  \(
  (a \circ (b \circ c))\vx
  \approx_\vx
  a(b(c\vx))
  \approx_\vx
  ((a \circ b) \circ c)\vx ,
  \)
  we have
  \[
  a \circ (b \circ c)
  =
  \ce(a \circ (b \circ c))
  =
  \ce((a \circ b) \circ c)
  =
  (a \circ b) \circ c .
  \]
  
  \noindent\ref{monoid2}.
  Since
  \(
  (\ci \circ a)\vx
  \approx_\vx
  a\vx
  \approx_\vx
  (a \circ \ci)\vx ,
  \)
  we have
  \[
    \ci \circ a
    =
    \ce(\ci \circ a)
    =
    \ce a
    =
    \ce(a \circ \ci)
    =
    a \circ \ci .
    \qedhere
  \]
\end{proof}
We recall the following construction from Definition~\ref{def:B}:
  \[
      A^\ast = \left\{ a \in \underline{A} \mid \ce a = a \right\} =
      \left\{ \ce a \in \underline{A} \mid a \in \underline{A}
    \right\}.
   \]
Define $\Unit \in A^\ast$ by $ \Unit = \ce \ci$.
\begin{proposition}
  \label{prop:FixEMonoid}
  The structure $(A^\ast, \circ, \Unit)$ is a monoid with unit
  $\Unit$.
\end{proposition}
\begin{proof}
  By Lemma \ref{monoid}\eqref{monoid1}, it suffices to show that
  $\Unit$ is a unit
  of $\circ$. For any $a \in A^\ast$, we
  have
    $
    \Unit \circ a 
    = \ci \circ a
    = \ce a
    = a
    $
    by Lemma \ref{monoid}\eqref{monoid2}.
    Similarly, 
    we have
    $
    a \circ \Unit = a.
    $
\end{proof}

\begin{definition}
  [{Koymans~\cite[Definition 6.3]{KoymansModelOfLambda}}]
  \label{def:CCM}
  Let $(M, \circ, I)$ be a monoid with unit $I$.
  \begin{enumerate}
    \item 
  $M$ \emph{has a paring} if it
  is equipped with elements $p,q \in M$
  and an operation
    $\pair{\cdot,\cdot}  \colon M \times M \to M$
  satisfying 
  \begin{enumerate}
   \item\label{def:CCMI}
     $p \circ \pair{a,b} = a$ and $q \circ \pair{a,b} = b$,
   \item\label{def:CCMII}
     $\pair{a,b} \circ c = \pair{a \circ c, b \circ c}$
  \end{enumerate}
  for each $a, b, c \in M$.

  \item
  $M$ is \emph{cartesian closed} if 
  it has a paring together with an element $\varepsilon \in M$
   and an operation
    $
    \LambdaAbst{\cdot}  \colon M \to M 
    $
  satisfying
  \begin{enumerate}[resume]
   \item\label{def:CCMIII}
     $\varepsilon \circ \pair{p,q} = \varepsilon$,
   \item\label{def:CCMIV}
     $\varepsilon \circ \pair{\LambdaAbst{a} \circ p, q} = a \circ
     \pair{p,q}$,
   \item\label{def:CCMV}
     $\LambdaAbst{\varepsilon} \circ \LambdaAbst{a} = \LambdaAbst{a}$,
   \item\label{def:CCMVI}
     $\LambdaAbst{\varepsilon \circ \pair{a \circ p, q}}
     =
     \LambdaAbst{\varepsilon} \circ a$
  \end{enumerate}
  for each $a \in M$.
  \end{enumerate}
\end{definition}
Coming back to the context of the monoid $(A^\ast, \circ, \Unit)$, 
define elements $p,q,\varepsilon \in A^\ast$
and operations
$\pair{\cdot,\cdot} \colon A^\ast\times A^\ast \to A^\ast$
and 
$\LambdaAbst{\cdot} \colon A^\ast \to A^\ast$ by
  \begin{align*}
    p &= \lambda^\dag \vx.\vx \ct,
    &
    q &= \lambda^\dag \vx. \vx \cf,
    &
    \varepsilon  &= \lambda^\dag \vx. \vx \ct (\vx \cf), \\
    \pair{a,b} &= \lambda^\dag \vx. [a\vx,b\vx],
    &
    \LambdaAbst{a}  &= \lambda^\dag \vx\vy. a[\vx,\vy].
  \end{align*}
Here, $\ct$, $\cf$, and $[\cdot,\cdot]$ are defined as in
\eqref{eq:Pairing} using $\lambda^{\dag}$ instead of $\lambda^\ast$.
\begin{theorem}[{cf.\ Koymans~\cite[Lemma 7.2]{KoymansModelOfLambda}}]
  \label{thm:CCM} 
  Let $A$ be a combinatory pre-model. 
  \begin{enumerate}
    \item \label{thm:CCM1}
  If $A$ is reflexive, then 
  the structure $(A^\ast, \circ, I, p,q, \pair{\cdot,\cdot})$ is a
  monoid with paring.
      \item \label{thm:CCM2}
  If $A$ is strongly reflexive, then the structure $(A^\ast, \circ,
  I, p,q,\varepsilon, \pair{\cdot,\cdot}, \LambdaAbst{\cdot})$ is a
  cartesian closed monoid.
  \end{enumerate}
\end{theorem}
\begin{proof}
\ref{thm:CCM1}.
We verify \eqref{def:CCMI}--\eqref{def:CCMII} of Definition \ref{def:CCM}.
Fix $a, b, c\in A^\ast$:
\begin{align*}
      (p \circ \pair{a,b})\vx
      &\approx_\vx
      p(\pair{a,b}\vx) 
      \approx_\vx
      p[a\vx,b\vx] 
      \approx_\vx
      [a\vx,b\vx]\ct
      \approx_\vx
      a\vx. %&& \text{(by Lemma \ref{lem:ProdObj}\eqref{lem:ProdObj1})}
      \\[.5em]
    (q \circ \pair{a,b})\vx &\approx_\vx b\vx.
    \quad \text{(similar to the above)}\\[.5em]
    % by Lemma \ref{lem:ProdObj}\eqref{lem:ProdObj2}.}
  %
      (\pair{a,b} \circ c)\vx
      &\approx_\vx
      \pair{a,b}(c\vx)
      \approx_\vx
      [a(c\vx),b(c\vx)]  \\
      &\approx_\vx
      [(a \circ c)\vx,(b \circ c)\vx]
      \approx_\vx
      \pair{a \circ c, b \circ c}\vx.
\end{align*}
Thus, the required equations follow from  Lemma~\ref{lem:reflexive}
and Lemma~\ref{lem:ReflE}.
\smallskip

\noindent\ref{thm:CCM2}.
We verify \eqref{def:CCMIII}--\eqref{def:CCMVI} of Definition \ref{def:CCM}.  Fix $a \in A^\ast$:
\begin{align*}
      (\varepsilon \circ \pair{p,q})\vx
      &\approx_\vx
      \varepsilon (\pair{p,q}\vx)
      \approx_\vx
      \varepsilon[p\vx,q\vx]
      \approx_\vx
      [p\vx,q\vx]\ct([p\vx,q\vx]\cf)\\
      &\approx_\vx
      p\vx(q\vx)
      %&& \text{(by Lemma \ref{lem:ProdObj})}
      \approx_\vx
      \vx\ct(\vx\cf)
      \approx_\vx
      \varepsilon \vx.\\[.5em]
      (\varepsilon\circ \pair{\LambdaAbst{a} \circ p, q})\vx
      &\approx_\vx
      \varepsilon(\pair{\LambdaAbst{a} \circ p, q}\vx)
      \approx_\vx
      \varepsilon[\LambdaAbst{a}(p\vx), q\vx]\\
      &\approx_\vx
      \varepsilon[\LambdaAbst{a}(\vx\ct), \vx\cf]
      \approx_\vx
      [\LambdaAbst{a}(\vx\ct), \vx\cf]\ct([\LambdaAbst{a}(\vx\ct),
      \vx\cf]\cf) \\
      &\approx_\vx
      \LambdaAbst{a}(\vx\ct)(\vx\cf) 
      %&& \text{(by Lemma \ref{lem:ProdObj})}
      \approx_\vx
      a[\vx\ct, \vx\cf]
      \approx_\vx
      a[p\vx, q\vx]\\
      &\approx_\vx
      a (\pair{p,q}\vx)
      \approx_\vx
      (a  \circ \pair{p,q})\vx. \\[.5em]
      (\LambdaAbst{\varepsilon} \circ \LambdaAbst{a} )\vx\vy
      &\approx_{\vx\vy}
      \LambdaAbst{\varepsilon}(\LambdaAbst{a}\vx)\vy
       \approx_{\vx\vy}
      \varepsilon[\LambdaAbst{a}\vx,\vy]\\
      &\approx_{\vx\vy}
      [\LambdaAbst{a}\vx,\vy]\ct([\LambdaAbst{a}\vx,\vy]\cf)
      %&& \text{(by Lemma \ref{lem:ProdObj})}
      \approx_{\vx\vy}
      \LambdaAbst{a}\vx\vy.\\[.5em]
      \LambdaAbst{\varepsilon \circ \pair{a \circ p, q}}\vx\vy
       &\approx_{\vx\vy}
      (\varepsilon \circ \pair{a \circ p, q})[\vx,\vy]
       \approx_{\vx\vy}
      \varepsilon [( a \circ p )[\vx,\vy], q[\vx,\vy]]\\
      & \approx_{\vx\vy}
      \varepsilon [( a ([\vx,\vy]\ct), [\vx,\vy]\cf]
      \approx_{\vx\vy}
      \varepsilon [a\vx, \vy]\\
      %&& \text{(by Lemma \ref{lem:ProdObj})}
      &\approx_{\vx\vy}
      \LambdaAbst{\varepsilon}(a \vx) \vy
      \approx_{\vx\vy}
      (\LambdaAbst{\varepsilon} \circ a)\vx\vy. 
\end{align*}
Since $A$ is strongly reflexive (i.e., $A[\vx]$ is reflexive),
we have the required equations as in the first case.
\end{proof}

Now, let $\Cat{C}_{A}$ be the Karoubi envelope of the monoid $(A^\ast,
\circ, I)$ seen as a single-object category (see
Koymans~\cite[Definition 4.1]{KoymansModelOfLambda};
Barendregt~\cite[5.5.11]{BarendregtLambda}). If $A$ is strongly
reflexive, then $\Cat{C}_{A}$ is a cartesian
closed category with a reflexive object $U = \Unit$ by
Theorem~\ref{thm:CCM} (see Koymans~\cite[Section
7]{KoymansModelOfLambda}).  Then, the homset $\Cat{C}_{A}(\terminal,
U)$, where $\terminal$ is a terminal object of $\Cat{C}_{A}$, has a
structure of a lambda algebra, and when $A$ is a lambda algebra,
$\Cat{C}_{A}(\terminal, U)$ is isomorphic to $A$ \cite[Sections 3 and
4]{KoymansModelOfLambda}.  For a strongly reflexive combinatory
pre-model $A$, in general, we would not have an isomorphism between
$A$ and $\Cat{C}_{A}(\terminal, U)$.  Instead, we obtain a certain
lambda algebra structure on $(\underline{A}, \cdot)$ induced by
$\Cat{C}_{A}(\terminal, U)$, which is analogous to the construction of
a lambda model from a combinatory model \cite[Section 6]{MeyerLambda}.
This is the subject of the next section.

%%%%%%%%%%%%%%%%%%%%%%%%%%%%%%%%%%%%%%%%%%%%%%%%%%%%%%%%%%%%%%%%%%
\section{Stability}\label{sec:Stability}
%%%%%%%%%%%%%%%%%%%%%%%%%%%%%%%%%%%%%%%%%%%%%%%%%%%%%%%%%%%%%%%%%%
It is known that lambda models are combinatory models which are stable
(Barendregt~\cite[5.6,3, 5.6.6]{BarendregtLambda}, Koymans~\cite[Section 1.4]{KoymansThesis}). Having seen that
strongly reflexive combinatory pre-models are the retracts of combinatory models in
Section~\ref{sec:AlgCombModel}, it is now straightforward to establish
an algebraic analogue of this fact for lambda algebras.

We begin with the following construction, which extends $\ce$ to
finitely many arguments. 
Note that the conventions of Notation \ref{not:CPM} and
Notation~\ref{rem:ClosedTerm} continue to apply.
\begin{definition}[Scott~\cite{ScottLambda}]
  \label{def:EpsilonN}
  For a combinatory pre-model $A$,  define $\varepsilon_n \in
  \underline{A}$ for each $n \geq 1$ inductively by
  \begin{align*}
    \varepsilon_1 &= \ce, &
    \varepsilon_{n+1} &= \cs (\ck \ce)(\cs(\ck
    \varepsilon_{n})).
  \end{align*}
\end{definition}
One can observe the following by straightforward calculation.
\begin{lemma}
  \label{lem:EpsilonN}
    Let $A$ be a combinatory pre-model.
    For each $n \geq 1$, and for each $m \in \Nat$  and $s,t \in
    \mathcal{T}(\left\{ \vx_1, \dots,\vx_m \right\} + \underline{A})$, 
    we have
    \[
      \varepsilon_{n+1}st  \approx_{\vx_1\dots\vx_{m}}
      \varepsilon_{n}(st). 
    \]
   In particular,  
  \begin{enumerate}
    \item\label{lem:EpsilonN1} $\varepsilon_n a \vx_1 \cdots \vx_n
  \approx_{\vx_1\dots\vx_n} a \vx_1 \cdots \vx_n$,
    \item\label{lem:EpsilonN2} $\varepsilon_n a \vx_1 \cdots \vx_{n-1}
      \approx_{\vx_1\dots\vx_{n-1}} \ce (a \vx_1 \cdots \vx_{n-1})$
  \end{enumerate}
  for each $n \geq 1$ and $a \in \underline{A}$.
\end{lemma}
If $A$ is a strongly reflexive combinatory pre-model, then the element $\varepsilon_n
a$ for each $a \in \underline{A}$ admits a succinct characterisation.
\begin{lemma}
  \label{lem:EspsilonNACM}
   If $A$ is strongly reflexive, then
   \[
     \varepsilon_n a 
     = \lambda^\dag \vx_1\dots \vx_n. a \vx_1\cdots \vx_n
   \]
   for each $a \in \underline{A}$ and $n \geq 1$.%
   \footnote{
  If $\varepsilon_{n}$ were defined by 
  $\varepsilon_1 = \ce$
  and $\varepsilon_{n+1} = \ce( \cs (\ck \ce)(\cs(\ck \varepsilon_{n})))$,
  we could even show 
     $
     \varepsilon_n 
     = \lambda^\dag \vy\vx_1\dots \vx_n. \vy \vx_1\cdots \vx_n.
     $
   }
\end{lemma}
\begin{proof}
  By straightforward induction $n \in \Nat$.
\end{proof}

The construction $\varepsilon_n$ provides yet another characterisation
of strong reflexivity.
\begin{proposition}
  A combinatory pre-model $A$ is strongly reflexive if and
  only if 
   \[
    a \vx_1 \cdots \vx_n
    \approx_{\vx_1\dots\vx_n} b \vx_1 \cdots \vx_n
     \implies
     \varepsilon_n a = \varepsilon_n b
   \]
   for each $n \geq 1$ and $a,b \in \underline{A}$.
\end{proposition}
\begin{proof}
  \noindent ($\Rightarrow$) 
  By the $n$-time applications of Proposition \ref{cor:AlgCombModelLamdaDag}
  and Lemma \ref{lem:EspsilonNACM}.
  \smallskip

  \noindent($\Leftarrow$) 
  We must show that $A[\vx]$ is reflexive. Let $t,u \in \mathcal{T}(\left\{
  \vx \right\} + \underline{A})$
  be such that $t \vy \approx_{\vx\vy} u \vy$. Then,
  $(\lambda^\dag \vx. t)\vx \vy
  \approx_{\vx\vy} (\lambda^\dag \vx. u) \vx \vy$, and so 
  $\varepsilon_2(\lambda^\dag \vx. t) 
  = \varepsilon_2(\lambda^\dag \vx. u)$.
  By Lemma \ref{lem:EpsilonN}, we have 
  $\ce t 
  \approx_\vx
  \ce((\lambda^\dag \vx. t)\vx) 
  \approx_\vx
  \ce((\lambda^\dag \vx. u)\vx) 
  \approx_\vx
  \ce u$.
\end{proof}

The notion of stability for combinatory models
also makes sense for strongly reflexive combinatory pre-models
(cf.\ Barendregt~\cite[5.6.4]{BarendregtLambda}; Meyer~\cite[Section
6]{MeyerLambda}; Scott~\cite{ScottLambda}).\footnote{The notion of
stability
also makes sense for combinatory pre-models in general. However, we
have not found any significant consequence of the notion in that
general setting.}
\begin{definition}
  \label{def:Stability}
  Let  $A = (\underline{A},\mathbin{\cdot},\ck,\cs,\ci,\ce)$ 
  be a strongly reflexive combinatory pre-model. Then, $A$
  is said to be \emph{stable} if 
  \begin{align}
    \label{eq:Stability}
    \ck &= \varepsilon_2 \ck, 
    &
    \cs &= \varepsilon_3 \cs, 
    &
    \ci &= \varepsilon_1 \ci, 
    &
    \ce &= \varepsilon_2 \ce.
  \end{align}
  A combinatory model is \emph{stable} if it is stable as a
  strongly reflexive combinatory pre-model (cf. Lemma \ref{lem:CombModAlg}).
\end{definition}

In a stable strongly reflexive combinatory pre-model, the constants $\ce$ and $\ci$
coincide with the usual construction of these constants from $\ck$ and $\cs$.
\begin{lemma}
  \label{lem:FixEI}
    If $A$
    is strongly reflexive and stable, then 
    $\ci = \cs \ck \ck$ and $\ce = \cs (\ck \ci)$.
\end{lemma}
\begin{proof}
  Suppose that $A$ is strongly reflexive  and stable. Then
  \begin{align*}
      \ci \vx 
      &\approx_\vx  \ck \vx (\ck \vx)
      \approx_\vx
      \cs \ck \ck \vx, \\[.5em]
      \ce \vx \vy
      &\approx_{\vx\vy}  
      \vx \vy 
      \approx_{\vx\vy}  
      \ck \ci \vy (\vx \vy )
      \approx_{\vx\vy}  
      \cs (\ck \ci) \vx \vy.
  \end{align*}
  Since $A$ is stable and reflexive, we have $\ci = \ce \ci = \ce (\cs \ck \ck) = \cs \ck
  \ck$. Similarly, we have $\ce = \cs (\ck \ci)$.
\end{proof}
The following is immediate from Lemma \ref{lem:EpsilonN}.
\begin{lemma}
  \label{lem:StabKSImpStabIE}
  If $A$
  is a combinatory pre-model  such that $\ci = \cs \ck\ck$ and $\ce =
  \cs (\ck \ci)$, then $\cs = \varepsilon_3 \cs$ implies $\ci =
  \varepsilon_1\ci$ and $\ce = \varepsilon_2\ce$.
\end{lemma}
\begin{remark}
  By the above two lemmas, the notion of stable combinatory models
  in the sense of Definition~\ref{def:Stability}
  agrees with the corresponding notion in the literature
  which does not include $\ci$ as a primitive
  \cite[5.6.4]{BarendregtLambda}.
\end{remark}

We recall the connection between stable combinatory models and lambda
models. As the definition of the latter, we adopt the following
characterisation.
\begin{definition}[{Barendregt~\cite[5.6.3]{BarendregtLambda}}]
  A \emph{lambda model} is a combinatory algebra $A = (\underline{A},
  \cdot, \ck,\cs)$ such that 
  the structure
  $(\underline{A}, \cdot, \ck,\cs,\ci,\ce)$, where
  $\ci = \cs\ck\ck$ and 
  $\ce = \cs(\ck\ci)$,
  is a combinatory model satisfying $\ck = \varepsilon_2 \ck$ and 
  $\cs = \varepsilon_3 \cs$.
\end{definition}
  The following is immediate from Lemma \ref{lem:FixEI} and Lemma
  \ref{lem:StabKSImpStabIE}.
\begin{proposition}[{Meyer\cite[Section 6]{MeyerLambda};
  Barendregt~\cite[5.6.6(i)]{BarendregtLambda}}]
  \label{prop:StableCMvxLM}
  The following are equivalent for a combinatory pre-model 
  $A$:
  \begin{enumerate}
    \item\label{prop:StableCMvxLM1} $A$ is a stable combinatory model.
    \item\label{prop:StableCMvxLM2} $(\underline{A},\cdot, \ck,\cs)$
      is a lambda model, and $\ci = \cs \ck \ck$ and $\ce = \cs (\ck \ci)$.
  \end{enumerate}
\end{proposition}

We establish an analogue of the above proposition
for strongly reflexive combinatory pre-models and lambda algebras. The following
characterisation of lambda algebras is often attributed to Curry.
\begin{definition}[{Barendregt~\cite[5.2.5, 7.3.6]{BarendregtLambda}}]
  \label{def:LambdaAlg}
    A \emph{lambda algebra} is a combinatory
    algebra
    \(
     A = (\underline{A},\mathbin{\cdot},\ck,\cs) 
    \) 
    satisfying the following equations:
    \begin{enumerate}
      \item\label{Ax:LambdaAlg1} 
        $\lambda^\ast \vx \vy. \ck \vx \vy = \ck$,

      \item\label{Ax:LambdaAlg2}
        $\lambda^\ast \vx \vy \vz. \cs \vx \vy  \vz = \cs$,

      \item\label{Ax:LambdaAlg3}
        $\lambda^\ast \vx \vy. \cs (\cs (\ck\ck)\vx)\vy 
        = \lambda^\ast \vx \vy \vz.  \vx \vz$,

      \item\label{Ax:LambdaAlg4}
        $\lambda^\ast \vx \vy \vz. \cs (\cs (\cs (\ck\cs)\vx)\vy) \vz
        = \lambda^\ast \vx \vy \vz.  \cs (\cs \vx \vz) (\cs \vy \vz)$,

      \item\label{Ax:LambdaAlg5}
        $\lambda^\ast \vx \vy. \cs (\ck \vx) (\ck \vy) 
        = \lambda^\ast \vx \vy.  \ck \vx \vy$.
    \end{enumerate}
\end{definition}
We recall the following fundamental result, which relates lambda algebras
to lambda models (cf.\ Theorem~\ref{thm:CombModel} and
Lemma~\ref{lem:CombModAlg}).
\begin{proposition}[{Meyer~\cite[Section~7]{MeyerLambda}}]
  \label{prop:LambdaAlgebraThm}
  \leavevmode
  \begin{enumerate}
    \item A combinatory algebra $A$ is a lambda algebra if and only if
      $A[\vX]$ is a lambda model.
    \item Every lambda model is a lambda algebra.
  \end{enumerate}
\end{proposition}
\begin{proof}
See Meyer~\cite[Section~7]{MeyerLambda}.
\end{proof}
\begin{remark}
  In Proposition~\ref{prop:LambdaAlgebraThm}, $A[\vX]$ denotes the
  polynomial algebra of $A = (\underline{A}, \cdot, \ck, \cs)$ as a
  combinatory algebra.  However, note that $A[\vX]$ coincides with the
  polynomial algebra of $A$ as a combinatory pre-model
  $(\underline{A},\cdot, \ck, \cs, \ci, \ce)$ where $\ci = \cs \ck
  \ck$ and $\ce = \cs (\ck \ci)$.  This is because the properties of
  $\ci$ and $\ce$ in $A[\vX]$ (as a combinatory pre-model) are
  derivable from those of $\ck$ and $\cs$.
\end{remark}

We can now establish the following correspondence (cf.\ Proposition
\ref{prop:StableCMvxLM}).
\begin{theorem}
  \label{thm:StableACMisLA}
  The following are equivalent for a combinatory pre-model
  $A$:
  \begin{enumerate}
    \item\label{thm:StableACMisLA1} $A$ is strongly reflexive and
      stable.
    \item\label{thm:StableACMisLA2} $(\underline{A},\cdot, \ck,\cs)$
      is a lambda algebra, and $\ci = \cs \ck \ck$ and $\ce = \cs (\ck \ci)$.
  \end{enumerate}
\end{theorem}
\begin{proof}
  Since the axioms of stability \eqref{eq:Stability} are closed
  equations,
  the following are equivalent by Theorem~\ref{thm:CombModel}, 
  Proposition \ref{prop:StableCMvxLM}, and  Proposition
  \ref{prop:LambdaAlgebraThm}:
  \begin{enumerate}
    \item $A$ is strongly reflexive and stable.
    \item $A[\vX]$ is a stable combinatory model.
    \item $(\underline{A[\vX]},\pa, \ck,\cs)$ is a
      lambda model, and $\ci \approx_{\vX} \cs \ck \ck$ and $\ce \approx_{\vX} \cs
      (\ck \ci)$.
    \item $(\underline{A}, \cdot, \ck,\cs)$ is a
      lambda algebra, and $\ci = \cs \ck \ck$ and $\ce = \cs (\ck \ci)$.
  \end{enumerate}
  Here, $\pa$ denotes the application of $A[\vX]$ (cf.\ Definition~\ref{def:PolyAlg}).
\end{proof}
We can also establish an algebraic analogue of the following result.
\begin{proposition}[{Meyer~\cite[Section 6]{MeyerLambda};
  Barendregt~\cite[5.6.6(ii)]{BarendregtLambda}}]
  \label{prop:LambdaModelThm}
  If $A$
  is a combinatory model, then $(\underline{A}, \cdot,
  \varepsilon_2 \ck,
  \varepsilon_3 \cs)$ is a lambda model.
\end{proposition}
\begin{proof}
    See Barendregt~\cite[5.6.6(ii)]{BarendregtLambda}.
\end{proof}
\begin{theorem}
  \label{thm:LambdaAlgThm}
  If $A$
  is a strongly reflexive combinatory pre-model, then $(\underline{A},
  \cdot, \varepsilon_2 \ck, \varepsilon_3 \cs)$ is a lambda algebra.
\end{theorem}
\begin{proof}
  If $A$ is strongly reflexive, then $A[\vX]$ is a
  combinatory model by Theorem \ref{thm:CombModel}. 
  Then, $(\underline{A[\vX]}, \pa, \varepsilon_2 \ck, \varepsilon_3 \cs)$
  is a lambda model by Proposition \ref{prop:LambdaModelThm}; hence it
  is a lambda algebra by Proposition \ref{prop:LambdaAlgebraThm}.
  Since $A$ is a retract of $A[\vX]$, we see that $(\underline{A},
  \cdot, \varepsilon_2 \ck, \varepsilon_3 \cs)$ is a lambda algebra.
\end{proof}
One can also check that the stable strongly reflexive combinatory
pre-model determined by the lambda algebra $(\underline{A}, \cdot,
\varepsilon_2 \ck, \varepsilon_3 \cs)$ is of the form $(\underline{A},
\cdot, \varepsilon_2 \ck, \varepsilon_3 \cs, \varepsilon_1 \ci,
\varepsilon_2 \ce)$.

\begin{remark}
  As we have noted at the end of Section~\ref{sec:CCCMoid}, a strongly
  reflexive combinatory pre-model $A$ determines a cartesian closed category with a reflexive
  object, which induces a lambda algebra structure on $(\underline{A},
  \cdot)$. One can verify
  that the constants $\ck$ and $\cs$ of the lambda algebra thus obtained
  coincide with $\varepsilon_2\ck$ and $\varepsilon_3 \cs$,
  respectively.
\end{remark}

We now understand that stable strongly reflexive combinatory pre-models and lambda
algebras are equivalent. However, compared with the five axioms of
lambda algebras (Definition \ref{def:LambdaAlg}),
we have seven axioms of strong reflexivity (Definition
\ref{def:AlgCombMod}) and four axioms of stability
\eqref{eq:Stability}. Nevertheless, since $\ci$ and
$\ce$ are definable from $\ck$ and $\cs$ in a stable strongly
reflexive combinatory pre-model, some of the axioms of
strong reflexivity and stability concerning the properties
of $\ci$ and $\ce$ become redundant. Specifically, 
consider the following conditions for a combinatory
pre-model 
\(
 A = (\underline{A},\mathbin{\cdot},\ck,\cs,\ci,\ce):
\)
 
\begin{enumerate}
  \myitem[CA]\label{def:CA} $\ci = \cs \ck \ck$ and $\ce = \cs (\ck \ci)$;
  \myitem[L1]\label{def:L1} $\ck = \varepsilon_2\ck$ and $\cs = \varepsilon_3\cs$;
  \myitem[L2]\label{def:L2} The equations \ref{def:AlgCombMod1}, \ref{def:AlgCombMod2},
    \ref{def:AlgCombMod5}, and \ref{def:AlgCombMod6} of Definition
    \ref{def:AlgCombMod}:
  \begin{itemize}
      \item[\ref{def:AlgCombMod1}.] 
        $\lambda^{\dag}\vx\vy. \ce(\cs(\cs(\ck\ck) \vx)\vy)
        = \lambda^\dag \vx\vy. \ce \vx$,
        
      \item[\ref{def:AlgCombMod2}.]
        $\lambda^{\dag}\vx\vy\vz. \ce(\cs(\cs(\cs(\ck \cs)\vx)\vy)\vz)
        = \lambda^\dag \vx\vy\vz. \ce (\cs(\cs \vx\vz)(\cs\vy\vz))$,

      \item[\ref{def:AlgCombMod5}.]
        $\lambda^{\dag}\vx\vy. \ce(\cs (\ck\vx) (\ck\vy))
        = \lambda^\dag \vx\vy. \ce (\ck(\vx\vy))$,

      \item[\ref{def:AlgCombMod6}.]
        $\lambda^{\dag}\vx. \ce(\cs (\ck\vx) \ci)
        = \lambda^\dag \vx. \ce \vx$.
  \end{itemize}
\end{enumerate}
By Lemma \ref{lem:FixEI}, the conditions \eqref{def:CA}, \eqref{def:L1},
and \eqref{def:L2} hold if $A$ is strongly reflexive and stable. Below, we show that these conditions are sufficient for
$A$ to be a stable strongly reflexive combinatory pre-model.
 
We fix a combinatory pre-model $A$ satisfying \eqref{def:CA},
\eqref{def:L1}, and \eqref{def:L2}.  First, \eqref{def:CA} and
\eqref{def:L1} imply
that $A$ is stable by Lemma \ref{lem:StabKSImpStabIE}. Thus, it
remains to show that $A$ is strongly reflexive.  To this
end, it suffices to show that $A$ is reflexive: for then $A[\vx]$ is
reflexive because \eqref{def:CA}, \eqref{def:L1}, and \eqref{def:L2} are
closed equations. 

To see that $A$ is reflexive, the following
observation is useful.
\begin{lemma}
\label{lem:LambdaAstLambdaDagAgree}
\eqref{def:CA}, \eqref{def:L1}, and \eqref{def:L2}\eqref{def:AlgCombMod6}
together imply $\lambda^\ast \vx. t = \lambda^\dag \vx. t$
for each $t \in \mathcal{T}(\left\{ \vx \right\} + \underline{A})$.
\end{lemma}
\begin{proof}
  By induction on the complexity of $t$. 
  The non-trivial case
  is $t \equiv (a,\vx)$ for $a \in \underline{A}$; in this case,
  we can use \eqref{def:L2}\eqref{def:AlgCombMod6}.
  For the other cases, the results follow from \eqref{def:CA},
  \eqref{def:L1}, and 
  Lemma~\ref{lem:EpsilonN}\eqref{lem:EpsilonN2}.
\end{proof}

\begin{lemma}
  \label{lem:minReflexive}
  $A$ is reflexive.
\end{lemma}
\begin{proof}
  By Proposition~\ref{lem:ReflexLambda} and Lemma
  \ref{lem:LambdaAstLambdaDagAgree}, it suffices to show that 
  $t \approx_{\vx} u$ implies
  $\lambda^\ast \vx. t  = \lambda^\ast \vx.u$  for each
  $t, u \in \mathcal{T}(\left\{ \vx \right\} + \underline{A})$. The proof is by
  induction on the derivation of $t \approx_{\vx} u$.
  Note that 
  we only need to show that the following defining equations of
  $A[\vx]$ are preserved by $\lambda^\ast \vx$:
  \begin{enumerate}
    \item $((\ck,s),t) \approx_{\vx} s$,
    \item $(((\cs,s),t),u) \approx_{\vx} ((s,u),(t,u))$,
    \item $(\ci,s) \approx_{\vx} s$,
    \item $((\ce,s),t) \approx_{\vx} (s,t)$,
    \item $(a,b) \approx_{\vx} ab$,
  \end{enumerate}
  where $a, b \in \underline{A}$ and $s,t,u\in
  \mathcal{T}(\left\{ \vx \right\} + \underline{A})$. By the assumption \eqref{def:CA}, equations
  \ref{eq:polyI} and \ref{eq:polyE} are derivable from \ref{eq:polyK} and
  \ref{eq:polyS}; thus we only need to deal with \ref{eq:polyK}, \ref{eq:polyS},
  and \ref{eq:polyInj}.

  We consider equation \ref{eq:polyK} as an example.
  First, note that $\ce (\lambda^\ast \vx. t)
  = \lambda^\ast \vx.t$ for each $t \in
   \mathcal{T}(\left\{ \vx \right\} + \underline{A})$ by \eqref{def:L1} and Lemma
  \ref{lem:EpsilonN}\eqref{lem:EpsilonN2}.
  Then, $\lambda^\ast \vx. ((\ck,s),t) = \lambda^\ast \vx. s$
  follows from \eqref{def:L2}\eqref{def:AlgCombMod1}
  by applying $\lambda^\ast \vx. s$  and $\lambda^\ast \vx. t$
  to both sides and removing some $\ce$ using Lemma
  \ref{lem:EpsilonN}\eqref{lem:EpsilonN2}.
  Similarly, one can show that equations \ref{eq:polyS} and
  \ref{eq:polyInj} are preserved by $\lambda^\ast \vx$ using
  \eqref{def:L2}\eqref{def:AlgCombMod2} and
  \eqref{def:L2}\eqref{def:AlgCombMod5}.
\end{proof}

From the above lemma, we obtain the desired conclusion.
\begin{theorem}
A combinatory pre-model is strongly reflexive and stable
if and only if it satisfies \eqref{def:CA}, \eqref{def:L1} and \eqref{def:L2}.
\end{theorem}

In the context of combinatory algebras, the conditions \eqref{def:L1}
and \eqref{def:L2} provide an alternative characterisation of lambda
algebras. These conditions are similar to Curry's axiomatisation 
in Definition \ref{def:LambdaAlg}, but they do not directly correspond to each other.
Instead, \eqref{def:L1} and \eqref{def:L2} naturally correspond to
Selinger's axiomatisation of lambda algebras~\cite{Selinger02},
as we will see below.

The relation between the two axiomatisations can be clarified by the
following characterisation of stable combinatory
models~\cite[5.6.5]{BarendregtLambda}. In the light of
Lemma~\ref{lem:FixEI}, the result can be stated as follows.
\begin{lemma}
  \label{lema:BarendregtStability}
  Let
  $A$
  be a combinatory model satisfying $\ci = \cs \ck \ck$ and
  $\ce = \cs (\ck \ci)$. Then, $A$ is stable if and only if it
  satisfies the following equations:
  \begin{align*}
    \ce \ck &= \ck, &
    \ce \cs &= \cs, &
    \ce (\ck a) &= \ck a, &
    \ce (\cs a) &= \cs a, &
    \ce (\cs a b ) &= \cs a b
  \end{align*}
  for each $a,b \in \underline{A}$.
\end{lemma}
\begin{proof}
  See Barendregt~\cite[5.6.5]{BarendregtLambda}.
\end{proof}

From the above lemma, we can derive the following algebraic analogue.
\begin{proposition}
  \label{prop:Stability}
  Let $A$
  be a strongly reflexive combinatory pre-model satisfying
  $\ci = \cs \ck \ck$ and  $\ce = \cs (\ck \ci)$.
  Then, $A$ is stable if and only if the following equations hold:
  \begin{align}
    \label{eq:SelingerStability}
    \ce \ck &= \ck, &
    \ce \cs &= \cs, &
    \ce (\ck \vx) &\approx_{\vx} \ck \vx, &
    \ce (\cs \vx) &\approx_{\vx} \cs \vx, &
    \ce (\cs \vx \vy) &\approx_{\vx\vy} \cs \vx \vy.
  \end{align}
\end{proposition}
\begin{proof}
  By Theorem~\ref{thm:CombModel}, $A$ is stable
  if and only if $A[\vX]$ is stable as a combinatory model. By Lemma
  \ref{lema:BarendregtStability}, the latter is equivalent to
  \begin{align*}
    \ce \ck &\approx_{\vX} \ck, &
    \ce \cs &\approx_{\vX} \cs, &
    \ce (\ck t) &\approx_{\vX} \ck t, &
    \ce (\cs t) &\approx_{\vX} \cs t, &
    \ce (\cs t u ) &\approx_{\vX} \cs t u
  \end{align*}
  for all $t,u \in \mathcal{T}(\vX + \underline{A})$, which in turn is equivalent to 
  \eqref{eq:SelingerStability}.
\end{proof}

Proposition~\ref{prop:Stability} allows us to see
the connection
between \eqref{def:L1} and \eqref{def:L2} and the following axiomatisation 
of lambda algebras by Selinger.
\begin{theorem}[{Selinger~\cite[Theorem~3]{Selinger02}}]
  \label{thm:Selinger}
  Let $A = (\underline{A}, \cdot, \ck,\cs)$ be a combinatory algebra. 
  Then, $A$ is a lambda algebra if and only if the following
  equations hold with $\ci = \cs
  \ck \ck$ and $\ce = \cs (\ck \ci)$:%
  \footnote{As noted by Selinger~\cite[Section~2.4]{Selinger02},
  condition \ref{thm:Selinger3} follows from \ref{thm:Selinger8}, so
it is redundant.}
  {
    \setlength{\columnsep}{-65pt}
  \begin{multicols}{2}
  \begin{enumerate}[label=(\alph*)]
    \item\label{thm:Selinger1}  
      $\ce \ck = \ck$,
    \item\label{thm:Selinger2} 
      $\ce \cs = \cs$,
    \item\label{thm:Selinger3} 
      $\ce (\ck \vx) \approx_{\vx} \ck \vx$,
    \item\label{thm:Selinger4} 
      $\ce (\cs \vx) \approx_{\vx} \cs \vx$,
    \item\label{thm:Selinger5} 
      $\ce (\cs \vx \vy) \approx_{\vx\vy} \cs \vx \vy$,
    \item\label{thm:Selinger6} 
      $\cs(\cs(\ck\ck) \vx)\vy
        \approx_{\vx\vy} \ce \vx$,
    \item\label{thm:Selinger7} 
      $\cs(\cs(\cs(\ck \cs)\vx)\vy)\vz
        \approx_{\vx\vy\vz}  \cs(\cs \vx\vz)(\cs\vy\vz)$,
    \item\label{thm:Selinger8} 
      $\cs (\ck\vx) (\ck\vy)
        \approx_{\vx\vy} \ck(\vx\vy)$,
    \item\label{thm:Selinger9} 
      $\cs (\ck\vx) \ci
        \approx_{\vx} \ce \vx$.
      \item[]% just to balance the list
  \end{enumerate}
  \end{multicols}
  }
\end{theorem}
\begin{proof}
  By Proposition~\ref{prop:Stability}, the conditions \eqref{def:L1} and
  \eqref{def:L2} imply the nine conditions above.
  Conversely, the proofs of Lemma \ref{lem:LambdaAstLambdaDagAgree} and Lemma \ref{lem:minReflexive} show that
  the conditions \ref{thm:Selinger3} and \ref{thm:Selinger5}--\ref{thm:Selinger9} imply that the structure $(\underline{A}, \cdot, \ck,
  \cs, \ci, \ce)$ with $\ci = \cs \ck \ck$ and $\ce =  \cs (\ck \ci)$
  is a strongly reflexive combinatory pre-model;
  hence they imply \eqref{def:L2}. Then, \eqref{def:L1} follows from
  \ref{thm:Selinger1}--\ref{thm:Selinger5}
  by Proposition \ref{prop:Stability}.
\end{proof}

The correspondence between \eqref{def:L1}--\eqref{def:L2} and
\ref{thm:Selinger1}--\ref{thm:Selinger9} are evident: in
\ref{thm:Selinger1}--\ref{thm:Selinger9} the conditions \eqref{def:L1}
is broken down into an equivalent set of smaller pieces and some $\ce$'s are
omitted from \eqref{def:L2} using \ref{thm:Selinger3} and \ref{thm:Selinger5}.

\subsection*{Acknowledgement}
We thank the anonymous referees for helpful comments and suggestions.
This research was supported by the Core-to-Core Program (A.\ Advanced
Research Networks) of the Japan Society for the Promotion of Science.

\bibliographystyle{abbrv}
\bibliography{ref.bib}
\end{document}